\documentclass[twocolumn, switch]{article} 

\usepackage{preprint}

\usepackage{amsmath, amsthm, amssymb, amsfonts}

\newcounter{ass}                           
\newcounter{rem}
\newcounter{lem}

\newcounter{def}

\newcounter{thm}

\newtheorem{assumption}[ass]{Assumption}
\newtheorem{remark}[rem]{Remark}
\newtheorem{lemma}[lem]{Lemma}

\newtheorem{definition}[def]{Definition}

\newtheorem{theorem}[thm]{Theorem}

\usepackage{appendix}

\usepackage[numbers,square]{natbib}
\usepackage[utf8]{inputenc}	
\usepackage[T1]{fontenc}	
\usepackage{xcolor}		
\usepackage[colorlinks = true,
            linkcolor = purple,
            urlcolor  = blue,
            citecolor = cyan,
            anchorcolor = black]{hyperref}	
\usepackage{booktabs} 		
\usepackage{nicefrac}		
\usepackage{microtype}		
\usepackage{lineno}		
\usepackage{float}			
\usepackage{multirow}
\usepackage{graphicx}
\graphicspath{ {./images/} }

\usepackage{lipsum}		

\usepackage{newfloat}
\DeclareFloatingEnvironment[name={Supplementary Figure}]{suppfigure}
\usepackage{sidecap}
\sidecaptionvpos{figure}{c}

\usepackage{titlesec}
\titlespacing\section{0pt}{12pt plus 3pt minus 3pt}{1pt plus 1pt minus 1pt}
\titlespacing\subsection{0pt}{10pt plus 3pt minus 3pt}{1pt plus 1pt minus 1pt}
\titlespacing\subsubsection{0pt}{8pt plus 3pt minus 3pt}{1pt plus 1pt minus 1pt}

\usepackage{colortbl}
\usepackage{hhline}
\definecolor{myblue}{RGB}{0,102,204}

\usepackage{tikz,xcolor,hyperref}

\definecolor{lime}{HTML}{A6CE39}
\DeclareRobustCommand{\orcidicon}{
	\begin{tikzpicture}
	\draw[lime, fill=lime] (0,0) 
	circle [radius=0.16] 
	node[white] {{\fontfamily{qag}\selectfont \tiny ID}};
	\draw[white, fill=white] (-0.0625,0.095) 
	circle [radius=0.007];
	\end{tikzpicture}
	\hspace{-2mm}
}
\foreach \x in {A, ..., Z}{\expandafter\xdef\csname orcid\x\endcsname{\noexpand\href{https://orcid.org/\csname orcidauthor\x\endcsname}
			{\noexpand\orcidicon}}
}

\title{Robust tracking MPC for perturbed nonlinear systems - Extended version}

\author{
	Marco Polver\\
	Dipartimento di Ingegneria Gestionale, dell'Informazione e della Produzione\\
	Università degli Studi di Bergamo (24044 Dalmine, Bergamo, Italy)\\
	\texttt{polvermarco@gmail.com}\\
	\And
	Daniel Limon\\
	Departamento de Ingeniería de Sistemas y Automática, Escuela Técnica Superior de Ingeniería\\Universidad de Sevilla (41092 Sevilla, Spain)\\
	\texttt{dlm@us.es}\\
	\And
	Fabio Previdi\\
	Dipartimento di Ingegneria Gestionale, dell'Informazione e della Produzione\\
	Università degli Studi di Bergamo (24044 Dalmine, Bergamo, Italy)\\
	\texttt{fabio.previdi@unibg.it}\\
	\And
	Antonio Ferramosca\\
	Dipartimento di Ingegneria Gestionale, dell'Informazione e della Produzione\\
	Università degli Studi di Bergamo (24044 Dalmine, Bergamo, Italy)\\
	\texttt{antonio.ferramosca@unibg.it}\\
}

\begin{document}

\twocolumn[ 
  \begin{@twocolumnfalse} 

    \maketitle

    \begin{abstract}
      This paper presents a novel robust predictive controller for constrained nonlinear systems that is able to track piece-wise constant setpoint signals. The tracking model predictive controller presented in this paper extends the nonlinear MPC for tracking to the more complex case of nonlinear systems subject to bounded and not necessarily additive perturbations. The optimal control problem that is solved at each step penalizes the deviation of the predicted nominal system trajectory from an artificial reference, which is added as a decision variable, as well as the distance between the artificial reference and the setpoint. Robust feasibility is ensured by imposing conservative constraints that take into account the effect of uncertainties and convergence to a neighborhood of any feasible setpoint is guaranteed by means of an appropriate terminal cost and an extended stabilizing terminal constraint. In the case of unreachable setpoints, convergence to a neighborhood of the optimal reachable steady output is also proved.
    \end{abstract}
    \vspace{0.35cm}

  \end{@twocolumnfalse} 
] 



\section{Introduction}
\label{sec:introduction}
Model Predictive Control (MPC) is an advanced control method for multivariable linear and nonlinear systems, that is able to provide stability and robustness, taking into account input, state, and output constraints \cite{Camacho2007,Rawlings2017_MPC}.
Most MPCs are designed for regulation purposes, namely they aim at steering the system state to a fixed setpoint. The asymptotic stability of the closed-loop system and the recursive feasibility of the MPC can be guaranteed by employing terminal cost functions and terminal constraints that satisfy some conditions in a neighborhood of the reference equilibrium \cite{Mayne2014model,Mayne2000constrained}. 

It is not unusual, however, for setpoints to change abruptly over time. In such cases, the design of the MPC for the previous setpoint might be invalid for the following one and the closed-loop system might fail to track the desired reference signal \cite{Bemporad1997nonlinear,Pannocchia2005offset,Rossiter1996guaranteeing}. Among the solutions to the issues provided by changing setpoints, we find dual-mode controllers, as the one proposed in \cite{Chisci2003dual}, MPC algorithms that guarantee local stability and asymptotic tracking of constant references like the one proposed in \cite{Magni2005solution}, and reference governors, e.g. \cite{Bemporad1997nonlinear,Angeli1999governors,Garone2017reference}. In cases where the evolution of the reference is known or predictable, another solution is given by offset-free MPC formulations, e.g. the one proposed in \cite{Kohler2022offset}, a formulation that simply minimizes the predicted output error, does not require complex offline computations and allows also for the output regularization of systems that are subject to (predictable) disturbances.

Nowadays, many tracking MPC formulations are inspired by the predictive controllers for linear systems presented in \cite{Limon2015_tracking,Ferramosca2009_tracking,Limon2008_linear_tracking}. The main feature of these methods is the addition of an artificial reference as a new decision variable; the cost function to be minimized penalizes both the tracking error with respect to the artificial steady state and the deviation between the artificial setpoint and the actual setpoint (by means of the so-called offset cost function); moreover, an extended terminal constraint based on an invariant set for tracking is added. All these features ensure the recursive feasibility of the controller and the convergence of the system output to the admissible steady output that is the closest to the actual setpoint. Following the same rationale, an MPC for tracking piece-wise constant reference signals for nonlinear systems was presented in \cite{Limon2018_NMPCT}. A further MPC for tracking for nonlinear systems was presented in \cite{Berberich2022mpc}, which controls perturbation-free nonlinear systems without hard state constraints exploiting the linearized dynamics of the system.

The previously listed MPCs for tracking, however, do not deal with the issues that might occur in case of bounded disturbances or model uncertainties, which could still lead to instability and infeasibility. In the case of MPCs for regulation, robustness is usually achieved by means of min-max or tube-based controllers. 
The robust MPC algorithms that follow a min-max formulation, as \cite{Cuzzola2002minmax,bemporad2003min,diehl2007formulation,Diehl2004minmax,Kerrigan2004_minmax_MPC,kothare1996robust,lazar2006min,Liu2018minmax,Villanueva2017minmax}, compute a control sequence that verifies the constraints for any possible perturbation, while minimizing the worst case performance index.

Tube-based MPC \cite{Mayne2005_robust_MPC} is a feedback MPC, as its cost function is minimized by optimizing over a class of control laws that is chosen so as to reduce the deviation of the perturbed state trajectories from the nominally predicted ones.
All the possible deviations generate a tube, i.e. a set (or sequence of sets) that describes how the uncertainty propagates while making nominal predictions and is usually explicitly used to restrict the state and input constraints on the nominal predictions to ensure robust feasibility of the controller. Differently from min-max controllers, tube-based MPC formulations optimize the performance index related to the nominal, i.e. perturbation-free, case.
Among the available regulation tube-based MPCs, we find formulations based on tubes that are pre-computed offline, thus often more conservative, as \cite{Mayne2005_robust_MPC,Chisci2001_disturbances,Mayne2011tube,Yu2013tube}, and other controllers that exploit a parameterized version of the tubes that allows them to be added as decision variables in the optimization problem that has to be solved online, as in \cite{Kohler2021_RNMPCT,Lopez2019dynamic,Rakovic2012homothetic,Rakovic2016elastic,Rakovic2012parameterized}.

In the tracking context, effective robust tube-based MPCs were proposed for the control of linear systems subject to additive disturbances in \cite{Ferramosca2012_robust_MPCT,Limon2010_robust_linear_MPC}, while a solution for nonlinear affine systems was presented in \cite{Zamani2022_robust_NMPCT} for tracking unknown references in case of slowly varying disturbances and a solution for nonlinear systems subject to additive perturbations can be found in \cite{Cunha2022rnmpct}.

The aim of the present work is to extend the results of \cite{Limon2018_NMPCT} to the control of nonlinear systems that are subject to generic not necessarily additive perturbations, in order to provide the design conditions that ensure recursive feasibility and robustness of the controller and input-to-state stability of the closed-loop system \cite{Limon2009iss}.
With this aim, the robust design of the controller is inspired by \cite{Limon2010_robust_NMPC}, where a robust MPC for the control of uncertain uniformly continuous systems is presented, while input-to-state stability is proved following similar arguments to those made in \cite{Kohler2020_periodic}, where a nonlinear tracking MPC for known periodic signals is introduced. However, while \cite{Kohler2020_periodic} treats perturbation-free systems and known periodic (and potentially unreachable) reference signals and imposes the employment of quadratic cost functions, the context of the current work is different; in fact, we consider systems subject to bounded perturbations, setpoints that change unpredictably and allow for the usage of a wider family of cost functions.

The manuscript is organized as follows. In Section 2, the constrained tracking problem is stated. In Section 3, the proposed robust MPC for tracking is presented and its convergence properties are proved. In Section 4, a case study showing the effectiveness of the proposed method is introduced. In Section 5, some conclusions are drawn.

\subsection*{Notation}
$\mathbb{R}_{> i}$ ($\mathbb{R}_{\geq i}$) denotes the set of real numbers that are greater than (or equal to) $i$.
Concatenations of column vectors are represented as $(a,b) \coloneqq  [a^\top \; b^\top]^\top$.
Bold lowercase letters denote sequences of scalars or vectors ($\mathbf{u} = \{ u(k),\hdots,u(j) \}$).
The $n \times n$ identity matrix is denoted $I_n$, while the diagonal matrix whose diagonal is equal to the vector $a$ is written as $\mathrm{diag}(a)$.
Calligraphic uppercase letters are used to indicate function classes ($\mathcal{K}$, $\mathcal{K}_\infty$ and $\mathcal{KL}$) and sets, e.g. $\mathcal{X}$, $\mathcal{U}$, and $\mathcal{T}$. The $j$-ary Cartesian power of a set $\mathcal{U}$, i.e. $\mathcal{U} \times \mathcal{U} \times \hdots \times \mathcal{U}$, is denoted $\mathcal{U}^j$. The absolute value of a scalar value $a$ is denoted as $\vert a \vert$. The 2-norm of the vector $x$ is denoted $\Vert x \Vert$.

The solution of system $x(k+1) = f \left( x(k),u(k),w(k) \right)$, where $x$ is the system state, $u$ the system input, and $w$ a perturbation, at sampling time $j$, starting from the initial condition $x$, given the input sequence $\mathbf{u} \coloneqq \left\{ u(0),\hdots,u(j-1) \right\}$ and the perturbation signal $\mathbf{w} \coloneqq \left\{ w(0),\hdots,w(j-1) \right\}$, is denoted $\phi(j;x,\mathbf{u},\mathbf{w})$. 
The projection of a set $\mathcal{A} \subset \mathcal{X} \times \mathcal{U}$ on the state dimensions is denoted $\mathrm{proj}_x(\mathcal{A})$.

A function $\alpha: \; \mathbb{R}_{\geq 0} \rightarrow \mathbb{R}_{\geq 0}$
is a $\mathcal{K}$ function if $\alpha(0) = 0$ and it is strictly increasing. 
A function $\alpha: \; \mathbb{R}_{\geq 0} \rightarrow \mathbb{R}_{\geq 0}$
is a $\mathcal{K}_\infty$ function if it is a $\mathcal{K}$ function and is unbounded.
A function $\beta: \; \mathbb{R}_{\geq 0} \times \mathbb{Z}_{\geq 0} \rightarrow \mathbb{R}_{\geq 0}$ is of class $\mathcal{KL}$ if for every fixed $k$ the function $\beta(\cdot,k)$ is a $\mathcal{K}$ function and for fixed $s$ the function $\beta(s,\cdot)$ is non-increasing and such that $\lim_{k \to \infty} \beta(s,k) = 0$.

\section{Problem formulation}
\label{sec:rnmpct_problem_formulation}
We consider a nonlinear time-invariant discrete-time system described by the equations
\begin{equation}
	\label{eq:rnmpct_system_equations}
	\begin{split}
		x(k+1) &= f\big(x(k),u(k),w(k)\big),\\
		y(k) &= h\big(x(k),u(k)\big),
	\end{split}
\end{equation}
where $x(k) \in \mathbb{R}^n$ is the system state, $u(k) \in \mathbb{R}^m$ is the controlled input, $y(k) \in \mathbb{R}^p$ is the system output, and $w(k) \in \mathcal{W} \subset \mathbb{R}^r$ is an unknown but bounded perturbation that lies in a known hypercube $\mathcal{W} \coloneqq \left\{ w: \; \vert w_i \vert \leq \bar{w}_i, \; \forall i \in [1,r] \right\}$. The system state $x$ is assumed to be measured at each sampling instant $k \geq 0$.
\begin{remark}
	The perturbation $w(k)$ can also represent bounded parameter variations over time or due to some scheduling variables. Therefore, in cases where such variations are difficult to track but are ensured to be bounded, the MPC proposed in the next sections can be effectively used to control nonlinear time-varying or parameter-varying systems.
\end{remark}
The system is subject to hard state and input constraints, namely
\begin{equation}
	\label{eq:rnmpct_system_constraints}
	\left( x(k),u(k) \right) \in \mathcal{Z} \coloneqq \mathcal{X} \times \mathcal{U}, \quad \forall k \geq 0.
\end{equation}
All the equilibrium points of the nominal system $(x_s,u_s,y_s)$ are such that
\begin{equation}
	\label{eq:rnmpct_state_input_equilibria}
	x_s = f(x_s,u_s,0), \quad	y_s = h(x_s,u_s).
\end{equation}
In order to ensure the recursive feasibility and stability of the proposed robust MPC for tracking, it is necessary to make some (non restrictive) assumptions on the functions $f(\cdot,\cdot,\cdot)$ and $h(\cdot,\cdot)$ and the sets $\mathcal{X}$ and $\mathcal{U}$.
\begin{assumption}[Continuity of the state and output functions]
	\label{ass:rnmpct_continuity}
	The output function $h(\cdot,\cdot)$ is continuous at any equilibrium point. The state function $f(\cdot,\cdot,\cdot)$ is such that each component $f_i(\cdot,\cdot,\cdot)$, $\forall i \in [1,n]$, is component-wise uniformly continuous. Consequently, there exist functions $\sigma_{x,ia}(\cdot),\sigma_{u,ib}(\cdot),\sigma_{w,ic}(\cdot) \in \mathcal{K}$ such that, for all $i \in [1,n]$,
	\begin{equation}
		\label{eq:rnmpct_componentwise_uniform_continuity}
		\begin{split}
			&\vert f_i(x,u,w) \! - \! f_i(\check{x},\check{u},\check{w}) \vert \leq \sum_{a = 1}^n \sigma_{x,ia} \left( x_a \! - \! \check{x}_a \right)\\ &+ \sum_{b = 1}^m \sigma_{u,ib} \left( \vert u_b \! - \! \check{u}_b \vert \right) + \sum_{c = 1}^r \sigma_{w,ic} \left( \vert w_c \! - \! \check{w}_c \vert \right),
		\end{split}
	\end{equation}
	for all $(x,u,w)$ and $(\check{x},\check{u},\check{w})$ in $\mathcal{Z} \times \mathcal{W}$.
\end{assumption}
\begin{assumption}[Properties of the constraint sets]
	\label{ass:rnmpct_constraint_set_properties}
	The sets $\mathcal{X}$ and $\mathcal{U}$ are closed and their interior is non-empty.
\end{assumption}
Assumption \ref{ass:rnmpct_constraint_set_properties} comes from standard MPC literature. On the contrary, Assumption \ref{ass:rnmpct_continuity} poses conditions on the state and output functions that are stricter than the usual continuity assumption provided in the MPC literature. Nonetheless, these conditions are needed to prove the input-to-state stability of the closed-loop system and are not particularly restrictive, as also clarified by the following remark.
\begin{remark}
	A particular case of component-wise uniformly continuous functions is that of component-wise Lipschitz continuous functions, which is a wide family of functions. 
	
	Indeed, in the common case in which the constraint set $\mathcal{Z}$ is compact, Lipschitz continuity is a condition that is easy to verify. Moreover, as proved in \cite{Manzano2019_Choki}, Lipschitz continuity implies also component-wise Lipschitz continuity.
\end{remark}
The correct operation of the here-proposed robust MPC for tracking requires the knowledge of a compact set of feasible setpoints with inactive constraints. In order to build it, we first define the restricted constraint set
\begin{equation}
	\label{eq:rnmpct_inactive_constraints}
	\hat{\mathcal{Z}} \coloneqq \left\{ z: \; z+\delta_z \in \mathcal{Z}, \; \forall \delta_z \; \mathrm{s.t.} \; \Vert \delta_z \Vert \leq \epsilon \right\},
\end{equation}
where $\epsilon > 0$ can be arbitrarily small. The set \eqref{eq:rnmpct_inactive_constraints} defines the feasible points $(x,u)$ with inactive constraints. The sets defining the equilibrium points where the constraints are not active, which are supposed to be non-empty, are therefore defined as
\begin{subequations}
	\begin{equation}
		\label{eq:rnmpct_state_input_equilibria_inactive}
		\mathcal{Z}_s \coloneqq \left\{ (x,u) \in \hat{\mathcal{Z}}: \; x = f(x,u,0) \right\},
	\end{equation}
	\begin{equation}
		\label{eq:rnmpct_output_equilibria_inactive}
		\mathcal{Y}_s \coloneqq \left\{ y = h(x,u): \; (x,u) \in \mathcal{Z}_s \right\}.
	\end{equation}
\end{subequations}
Given that tracking control problems require to stabilize the equilibrium $(x_s,u_s)$ related to a certain output reference $y_s$ (or the closest admissible equilibrium in case the setpoint were infeasible), we must avoid any ambiguity and limit ourselves to the control of systems for which it is impossible to find multiple equilibria providing the same output. Therefore, we make the following assumption on the output function $h(\cdot,\cdot)$ and the set $\mathcal{Y}_s$.
\begin{assumption}[Properties of the output function]
	\label{ass:rnmpct_existence_gx_gu}
	The system output is such that each steady output $y_s$ uniquely defines an equilibrium point $(x_s,u_s)$. Moreover, there exist a locally Lipschitz continuous function $g_x: \mathcal{Y}_s \to \mathbb{R}^n$ and a continuous function $g_u: \mathcal{Y}_s \to \mathbb{R}^m$ such that
	\begin{equation}
		\label{eq:rnmpct_gx_gu}
		x_s = g_x(y_s), \; u_s = g_u(y_s).
	\end{equation}
\end{assumption}
\begin{remark}
	A sufficient condition for Assumption \ref{ass:rnmpct_existence_gx_gu} to be verified is that the functions $f(x,u,0)$ and $h(x,u)$ are continuously differentiable and that the Jacobian matrix
	\begin{equation}
		\begin{bmatrix}
			(A(x_s,u_s)-I_n) & B(x_s,u_s) \\ C(x_s,u_s) & D(x_s,u_s)
		\end{bmatrix},
	\end{equation}
	where
	
	\begin{equation}
		\begin{aligned}
			A(x_s,u_s)\! &= \! \frac{\partial f(x,u,0)}{\partial x}(x_s,u_s), & \!\!\!\! B(x_s,u_s) \! &= \! \frac{\partial f(x,u,0)}{\partial u}(x_s,u_s),\\
			C(x_s,u_s)\! &= \! \frac{\partial h(x,u)}{\partial x}(x_s,u_s), & \!\!\!\! D(x_s,u_s) \! &= \! \frac{\partial h(x,u)}{\partial u}(x_s,u_s),
		\end{aligned}
	\end{equation}
	
	is non-singular for all $(x_s,u_s) \in \mathcal{Z}_s$.
\end{remark}
The aim of this work is to design a state-feedback robust MPC for tracking based on nominal predictions that is able to stabilize the closed-loop system and to steer its output to a bounded set containing the admissible steady output that is the closest to the actual setpoint $y_t$. Furthermore, this property must hold even when the setpoint $y_t$ is changed to a new constant value, i.e. $y_t$ is a piece-wise constant reference signal.

\section{Controller design}
\label{sec:rnmpct_controller_design}
In the following, our robust nonlinear MPC for tracking is presented and its robustness and stability properties are proved.

\subsection{Equilibrium-dependent feedback policies}
The here-proposed MPC belongs to the family of tube-based MPCs, and the size of its tube depends on the functions $\sigma_{x,ia}(\cdot)$ and $\sigma_{w,ia}(\cdot)$, which define how uncertainty is propagated in each state dimension. For this reason, consistently with the rationale of feedback MPC \cite{Limon2009iss}, we exploit a family of prestabilizing feedback policies
\begin{equation}
	\label{eq:rnmpct_feedback_policies}
	u(k) = \pi(y_s,x(k),v(k)),
\end{equation}
where $v(k) \in \mathbb{R}^s$, with the aim of reducing the effect of uncertainty propagation. Nonetheless, in order to keep the component-wise uniform continuity of the closed-loop system under the feedback policies $\pi(y_s,x,v)$, these have to fulfill a uniform continuity assumption.
\begin{assumption}[Continuity of the feedback policies and existence of a continuous function $g_v: \mathcal{Y}_s \to \mathbb{R}^s$]
	\label{ass:rnmpct_pi_continuity}
	For a fixed $y_s$, the feedback policies $\pi(y_s,x,v)$ are uniformly continuous with respect to $x$ and $v$. Furthermore, for each equilibrium point, there exists a unique $v_s$ such that $u_s = \pi(y_s,x_s,v_s)$. There also exists a continuous function $g_v: \mathcal{Y}_s \to \mathbb{R}^s$, such that $v_s = g_v(y_s)$.
\end{assumption}
	\begin{remark}
		Assumption \ref{ass:rnmpct_pi_continuity} gives the designer the possibility to design different feedback policies for different $y_s$ values. Nonetheless, for systems that do not show critical nonlinearities it could be convenient to design a unique feedback policy.
\end{remark}
Within this feedback framework, the task of our MPC becomes that of finding the optimal value for $y_s$ and $v(k)$. For this reason, we redefine the constraint sets $\mathcal{Z}$ and $\hat{\mathcal{Z}}$ as
\begin{equation}
	\label{eq:rnmpct_new_system_constraints}
	\mathcal{Z}_\pi \coloneqq \left\{ (x,v): \; \exists y_s \in \mathcal{Y}_s \; \mathrm{s.t.} \; \left( x,\pi \left( y_s,x,v \right) \right) \in \mathcal{Z} \right\},
\end{equation}
and
\begin{equation}
	\label{eq:rnmpct_new_hat_Z}
	\hat{\mathcal{Z}}_\pi \coloneqq \left\{ (x,v): \; (x,v)+ \delta_z \in \mathcal{Z}_\pi, \; \forall \delta_z \; \mathrm{s.t.} \; \Vert \delta_z \Vert \leq \epsilon \right\}.
\end{equation}
The new sets of the admissible steady states and outputs are obtained as projections of the following set:
\begin{align}
	\label{eq:rnmpct_new_equilibria_set}
	\mathcal{E}_{\pi,s} \coloneqq \Big\{ (x,v,y_s): \; &(x,v) \in \hat{\mathcal{Z}}_\pi, \; x = f \left( x,\pi \left( y_s,x,v \right),0 \right), \notag \\ &y_s = h \left( x,\pi \left( y_s,x,v \right) \right) \Big\}.
\end{align}
In particular, 
\begin{equation}
	\label{eq:rnmpct_new_state_input_equilibria}
	\mathcal{Z}_{\pi,s} \coloneqq \left\{ (x,v): \; \exists (x,v,y) \in \mathcal{E}_{\pi,s} \right\},
\end{equation}
\begin{equation}
	\label{eq:rnmpct_new_output_equilibria}
	\mathcal{Y}_{\pi,s} \coloneqq \left\{ y_s: \; \exists (x,v,y_s) \in \mathcal{E}_{\pi,s} \right\}.
\end{equation}
The previously defined sets play an essential role in the design of the proposed robust MPC for tracking.

\subsection{Nominal model predictive control for tracking}
As for other nonlinear tracking MPC formulations, i.e. \cite{Ferramosca2009_tracking,Limon2008_linear_tracking,Limon2018_NMPCT,Ferramosca2012_robust_MPCT,Kohler2020_periodic,Limon2010_robust_linear_MPC,Zamani2022_robust_NMPCT}, the here-proposed MPC is given the ability to steer the system output to a set containing the admissible steady output that is the closest to the real setpoint $y_t$, and to keep its feasibility properties also when $y_t$ changes. This is done by adding an artificial setpoint $y_s$ as an extra decision variable, by designing a cost function that includes a term $V_O(y_s - y_t)$, which penalizes the deviation between the artificial reference and the actual setpoint, and by designing an invariant set for tracking to be used as terminal constraint. Next, the fundamental ingredients of the optimization problem that has to be solved at each sampling instant are presented and explained.

For a given state $x$ and a given setpoint $y_t$, and employing a prediction horizon $N_p > 0$, the cost function of the proposed MPC is defined as
\begin{equation}
	\label{eq:rnmpct_cost_function}
	\begin{split}
		&V_{N_p} \big( x(k),y_t,\mathbf{v}(k),y_s(k) \big)\\ &\coloneqq  \sum_{j = 0}^{N_p-1} \ell \big( \hat{x}(j \vert k) - x_s(k),v(j \vert k) - v_s(k) \big)\\
		&+ V_f \big( \hat{x}(N_p \vert k) - x_s(k), y_s(k) \big) + V_O(y_s(k) - y_t)
	\end{split}
\end{equation}
where $\mathbf{v}(k) \coloneqq \left\{ v(0 \vert k), \hdots,v(N_p-1 \vert k) \right\}$ is the sequence of control inputs to be optimized, $v_f(y_s(k))$ is an input such that $u = \pi \left( y_s(k),x,v_f(y_s(k)) \right)$, and whose properties will be explained later in the text, $y_s(k)$ is the artificial reference, such that $x_s(k) = g_x(y_s(k))$, $u_s(k) = g_u(y_s(k))$, and $v_s(k) = g_v(y_s(k))$, while $\hat{x}(j \vert k)$, for $j=1,\hdots,N_p$, are nominal state predictions, i.e. 
\begin{equation}
	\hat{x}(j+1 \vert k) = f \left( \hat{x}(j \vert k),\pi\left(y_s(k),\hat{x}(j \vert k),v(j \vert k) \right),0 \right).
\end{equation}
The stage cost function $\ell(x - x_s,v - v_s)$ and the terminal cost function $V_f(x - x_s,y_s)$ penalize the tracking error with respect to the artificial reference $y_s$, while $V_O(y_s - y_t)$ penalizes the distance between the artificial reference $y_s$ and the actual setpoint $y_t$.

The optimization problem $P_{N_p}(x(k),y_t)$ that is solved at each sampling instant by the here-proposed MPC is the following:
\begin{subequations}
	\label{eq:rnmpct_nominal_optimization_problem}
	\begin{align}
		&\min_{(\mathbf{v}(k),y_s(k))} \; V_{N_p}(x(k),y_t,\mathbf{v}(k),y_s(k))\\
		&\mathrm{subject \; to:} \nonumber \\
		&\hat{x}(0 \vert k) = x(k), \\
		&\begin{aligned}
			\hat{x}(j+1 \vert k) = &f \left( \hat{x}(j \vert k),\pi \left( y_s(k),\hat{x}(j \vert k),v(j \vert k) \right),0 \right),\\ &\forall j\in [0,N_p-1]
		\end{aligned}\\
		&\big( \hat{x}(j \vert k),v(j \vert k) \big) \in \mathcal{Z}_\pi(j), \; \forall j\in [0,N_p-1]\\
		&x_s = g_x(y_s(k)), \; v_s=g_v(y_s(k)), \\
		&\big( \hat{x}(N_p \vert k),y_s(k) \big) \in \mathcal{T},
	\end{align}
\end{subequations}
where the set sequence $\left\{ \mathcal{Z}_\pi(j) \right\}_{\geq 0}$ and the terminal constraint set $\mathcal{T}$ will be defined later in the text. The optimal solution to this optimization problem and the optimal cost function are denoted as $\big( \mathbf{v}^0(x(k),y_t), y_s^0(x(k),y_t) \big)$, or $(\mathbf{v}^0(k),y_s^0(k))$ for brevity, and $V_{N_p}^0(x(k),y_t)$, respectively. Considering the receding horizon policy, the control law provided by the proposed MPC is defined as
\begin{equation}
	\kappa_{N_p}(x(k),y_t)   \coloneqq   \pi   \left(y_s^0(k),  x(k),  v^0(0 \vert k) \right),
\end{equation}
where $v^0(0 \vert k)$ is the first element of $\mathbf{v}^0(x(k),y_t)$. Note that an extended terminal constraint is imposed on both the terminal state $\hat{x}(N_p \vert k)$ and the artificial reference $y_s$, and that none of the constraints involved in $P_{N_p}(x(k),y_t)$ depends on the actual setpoint $y_t$, which leads to the existence of a region of attraction $\mathcal{X}_{N_p}$ where $P_{N_p}(x(k),y_t)$ is feasible for all $y_t \in \mathbb{R}^p$.

\subsection{Robust design of the controller}
The robustness of our MPC for tracking is ensured by imposing that the nominal system trajectory stays inside the set sequence $\left\{ \mathcal{Z}_{\pi}(j) \right\}_{j \geq 0}$ of tightened constraints \cite{Limon2009iss,Limon2010_robust_NMPC}. Before explaining in detail how to compute such sequence, we introduce the following notation: with $\phi_\pi \left( j;x,y_s,\mathbf{v},\mathbf{w} \right)$ we denote the solution of system \eqref{eq:rnmpct_system_equations} at sampling time $j$ given the initial state $x$, the feedback policy $\pi(y_s,x,v)$, where $v$ comes from the sequence $\mathbf{v}$, and the uncertainty signal $\mathbf{w}$.

The computation of $\left\{ \mathcal{Z}_{\pi}(j) \right\}_{j \geq 0}$ requires the existence of the set sequences $\left\{ \mathcal{F}(j) \right\}_{j \geq 0}$ and $\left\{ \mathcal{R}(j) \right\}_{j \geq 0}$, whose properties are outlined in the following two assumptions.

\begin{assumption}[Bounds on the effect of the initial condition on the state trajectory]
	\label{ass:rnmpct_F_sequence}
	The sequence $\{ \mathcal{F}(j) \}_{j \geq 0}$ is such that:
	\begin{itemize}
		\item $\mathcal{F}(0) \coloneqq \{ x \in \mathbb{R}^n: \vert x_i \vert \leq \sum_{a=1}^r \sigma_{w,ia} (\bar{w}_a), \forall i \in [1,n] \}$;
		\item $\forall j > 0$, the sets $\mathcal{F}(j)$ ensure that, for all $y_s \in \mathcal{Y}_{\pi,s}$, for every feasible $(x,\mathbf{v})$, and for every $\check{x}$ such that $\check{x}-x \in \mathcal{F}(0)$, the condition $$\phi_\pi\big(j;\check{x},y_s,\mathbf{v},\mathbf{0}\big) \allowbreak \in \phi_\pi\big(j;x,y_s,\mathbf{v},\mathbf{0}\big) \oplus \mathcal{F}(j)$$ is verified.
	\end{itemize}
\end{assumption}
\begin{assumption}[Bounds on the effect of the perturbations on the state trajectory]
	\label{ass:rnmpct_R_sequence}
	The set sequence $\{ \mathcal{R}(j) \}_{j \geq 0}$ is such that, defining $\mathcal{R}(0) \coloneqq \{0\}$, it ensures that for all $y_s \in \mathcal{Y}_{\pi,s}$:
	\begin{enumerate}
		\item for every feasible $(x,\mathbf{v})$, the sets $\mathcal{R}(j)$, for $j > 0$, are such that $$\phi_\pi\left(j;x,y_s,\mathbf{v},\mathbf{w}\right) \in \\ \phi_\pi\left(j;x,y_s,\mathbf{v},\mathbf{0}\right) \oplus \mathcal{R}(j) \; \forall \mathbf{w} \in \mathcal{W}^j;$$
		\item $\mathcal{F}(j) \oplus \mathcal{R}(j) \subseteq \mathcal{R}(j+1)$, $\forall j \geq 0$.
	\end{enumerate}
\end{assumption}
The existence of $\left\{ \mathcal{F}(j) \right\}_{j \geq 0}$ ensures that, starting from close initial conditions, the distance between two different trajectories is bounded, while $\left\{ \mathcal{R}(j) \right\}_{j \geq 0}$ bounds the distance between the nominal trajectories computed by our MPC and the possible trajectories of the perturbed system.
\begin{remark}
	Note that the conditions of the sequences $\left\{ \mathcal{F}(j) \right\}_{j \geq 0}$ and $\left\{ \mathcal{R}(j) \right\}_{j \geq 0}$ are required to hold for all the feedback policies $\pi(y_s,x,v)$. 
	
	A method to compute these sequences in the particular case of component-wise Lipschitz continuous systems is reported in the Appendix \ref{alg:F_R}.
\end{remark}
Given the sequence $\left\{ \mathcal{R}(j) \right\}_{j \geq 0}$, the tightened constraint set sequence $\left\{ \mathcal{Z}_\pi(j) \right\}_{j \geq 0}$ is defined as
\begin{equation}
	\label{eq:rnmpct_restricted_constraints}
	\mathcal{Z}_\pi(j) \coloneqq \mathcal{Z}_\pi \ominus \left( \mathcal{R}(j) \times \{0\} \right).
\end{equation}
Now that the tightened constraints introduced in the control problem $P_{N_p}(\cdot,\cdot)$ have been defined, let us introduce the set of feasible setpoints $\mathcal{Y}_t$ as
\begin{equation}
	\label{eq:rnmpct_feasible_setpoints}
	\begin{split}
		\mathcal{Y}_t   \coloneqq   \Big\{ y_s   \in   \mathcal{Y}_{\pi,s} :  \big(g_x(y_s),g_v(y_s)\big)   \in   \mathcal{Z}_\pi(N_p) \Big\}.
	\end{split}
\end{equation}

\subsection{Recursive feasibility and convergence to a neighborhood of the best admissible steady output}
In order to prove recursive feasibility of $P_{N_p}(\cdot,\cdot)$, closed-loop stability, and convergence of the closed-loop system output to a neighborhood of the best admissible setpoint, the stage cost function, the terminal ingredients and the offset cost function must verify some conditions, which are listed in the following assumptions.

In order to ensure the recursive feasibility of $P_{N_p}(\cdot,\cdot)$, the terminal ingredients $\mathcal{T}$ and $V_f(\cdot)$ must satisfy the following assumption.
\begin{assumption}[Terminal ingredient properties]
	\label{ass:rnmpct_robust_terminal_set}
	The inputs $v_f(y_s) \in \mathbb{R}^s$, the set $\mathcal{T} \coloneqq \Big( \bigcup_{y_s \in \mathcal{Y}_t} \mathcal{X}_f(y_s) \Big) \times \mathcal{Y}_t$ and the sets $\Omega(y_s)$ are such that:
	\begin{enumerate}
		\item $\Omega(y_s)$ and $\mathcal{X}_f(y_s)$ are positively invariant sets for the nominal system $\hat{x}(k+1) = f \left( \hat{x}(k),\pi \left( y_s,\hat{x}(k),v_f(y_s) \right),0 \right)$;
		\item $\Omega(y_s) \times \{v_f(y_s)\} \subseteq \mathcal{Z}_\pi(N_p-1)$;
		\item $\mathcal{X}_f(y_s) \times \{v_f(y_s)\} \subseteq \mathcal{Z}_\pi(N_p)$;
		\item $\mathcal{X}_f(y_s) \oplus \mathcal{F}(N_p-1) \subseteq \Omega(y_s)$;
		\item $f \left( \hat{x},\pi \left( y_s,\hat{x},v_f(y_s) \right),0 \right) \in \mathcal{X}_f(y_s)$, for all $\hat{x} \in \Omega(y_s)$.
	\end{enumerate}
	Moreover, the terminal cost function $V_f(\cdot,\cdot)$ is a CLF such that, for all $y_s \in \mathcal{Y}_t$ and for all $x \in \Omega(y_s)$, it ensures
	\begin{align}
		\label{eq:rnmpct_lyap_ineq}
		&V_f \big( f\left( x,\pi \left( y_s,x,v_f(y_s) \right),0 \right) - x_s,y_s \big) - V_f(x - x_s,y_s)\notag \\ &\leq -\ell \left( x-x_s,v_f(y_s)-v_s \right).
	\end{align}
\end{assumption}
The conditions given by Assumption \ref{ass:rnmpct_robust_terminal_set} ensure that the set $\mathcal{T}$ is a \textit{robust invariant set for tracking}, the definition of which is provided hereafter.
\begin{definition}[Robust invariant set for tracking]
	For a given set of constraints $\mathcal{Z}_\pi$, a set of feasible setpoints $\mathcal{Y}_t \subseteq \mathcal{Y}_{\pi,s}$, and a local control law $v = \kappa(x,y_s)$, a set $\mathcal{T} \subset \mathbb{R}^n \times \mathbb{R}^s$ is an (admissible) robust invariant set for tracking for the system $x^+ = f (x,\pi(y_s,x,v),w)$ if, for all $(x,y_s) \in \mathcal{T}$, we have that $(x,\kappa(x,y_s)) \in \mathcal{Z}_\pi$, $y_s \in \mathcal{Y}_t$, and $(f (x,\pi(y_s,x,\kappa(x,y_s)),w),y_s) \in \mathcal{T}$, $\forall w \in \mathcal{W}$.
\end{definition}
In detail, the robustness of $\mathcal{T}$, which is used as terminal inequality constraint on both the terminal state $\hat{x}(N_p \vert k)$ and the artificial reference $y_s$, is with respect to the terminal state. Indeed, the sets $\Omega(y_s)$ are required to be robustly invariant for the system $x(k + 1)  =  f \big( x(k), \allowbreak \pi(y_s, x(k), v_f(y_s) ), 0 \big) + e$, where $e \in \mathcal{F}(N_p-1)$, while $\mathcal{Y}_t$, i.e. the set of feasible setpoints, is required to fulfill the typical conditions of a standard invariant set for tracking. Assumption \ref{ass:rnmpct_robust_terminal_set} also requires to design control inputs $v_f(y_s)$ that are able to locally asymptotically stabilize the system to any steady state contained in $\mathcal{Z}_{\pi,s}$.

Next, we state the conditions that the cost function must verify to enable the proof of recursive feasibility and input-to-state stability of the closed-loop system.
\begin{assumption}[Stage cost function and terminal cost function bounds]
	\label{ass:rnmpct_tracking_bounds}
	The stage cost function $\ell(\cdot,\cdot)$ and the terminal cost function $V_f(\cdot,\cdot)$ meet the next conditions:
	\begin{enumerate}
		\item There exists a function $\alpha_\ell(\cdot) \in \mathcal{K}_\infty$ such that $\ell(x,v) \geq \alpha_\ell(\Vert x \Vert)$ for all $(x,v) \in \mathbb{R}^n \times \mathbb{R}^s$.
		The stage cost function $\ell(\cdot,\cdot)$ is uniformly continuous in $\mathcal{Z}_\pi$. Therefore, there exist functions $\sigma_{\ell,x}(\cdot),\sigma_{\ell,v}(\cdot) \in \mathcal{K}$ such that
		\begin{equation}
			\begin{split}
				\vert \ell(x,v) - \ell(\check{x},\check{v}) \vert &\leq \sigma_{\ell,x}(\Vert x - \check{x} \Vert) + \sigma_{\ell,v}(\Vert v - \check{v} \Vert),\\ &\forall (x,v),(\check{x},\check{v}) \in \mathcal{Z}_\pi.
			\end{split}
		\end{equation}
		\item There exists a quadratic function $\alpha_f \left( \Vert x \Vert \right) \coloneqq \bar{\alpha}_f \Vert x \Vert^2$ such that $V_f \left( x - x_s, y_s \right) \leq \bar{\alpha}_f \Vert x-x_s \Vert^2$ for all $y_s \in \mathcal{Y}_t$ and for all $x \in \Omega(y_s)$.
		\item The terminal cost function $V_f(\cdot,\cdot)$ is uniformly continuous with respect to its first argument, namely there exists a function $\sigma_f(\cdot) \in \mathcal{K}$ such that
		\begin{equation}
			\vert V_f(x - x_s,y_s) - V_f(\check{x} - x_s,y_s) \vert \leq \sigma_f \left( \Vert x - \check{x} \Vert \right).
		\end{equation}
	\end{enumerate}
\end{assumption}
\begin{remark}
	The function $V_f(x-x_s,y_s)$ is not required to be continuous with respect to $y_s$. This allows the designer to choose different terminal cost functions for different equilibria or groups of equilibria.
	Nonetheless, Assumption \ref{ass:rnmpct_tracking_bounds} requires the functions $V_f(x-x_s,y_s)$ to share an upper bound and a function $\sigma_f(\cdot)$ defining a common uniform continuity property. Of course, when it is possible, a simple way to verify such conditions is to use a unique terminal cost function for all the feasible setpoints $y_s \in \mathcal{Y}_t$. An algorithm for the design of the pre-stabilizing control laws $\pi(\cdot,\cdot,\cdot)$, the terminal cost functions $V_f(\cdot,\cdot)$ and the robust invariant set for tracking $\mathcal{T}$ is provided in the Appendix \ref{alg:terminal_ingredients}.
	
	If the nonlinear system under control can be described with an exact quasi-LPV realisation based on some scheduling parameters $\xi(t)$, then control laws $\pi(\cdot,\cdot,\cdot)$ and terminal ingredients $\big( V_f(\cdot,\cdot),\mathcal{T} \big)$ verifying the assumptions can be designed following LMI-based methods like those used in \cite{Kohler2020qlpv,Morato2021qlpv}. 
\end{remark}
In order to ensure that the closed-loop system output converges to a neighborhood of the best admissible steady output, it is fundamental that the offset function $V_O(\cdot)$ and the set $\mathcal{Y}_t$ verify the following conditions.
\begin{assumption}[Properties of $V_O(\cdot)$ and $\mathcal{Y}_t$]
	\label{ass:rnmpct_offset_function}
	\begin{enumerate}
		\item The set of feasible setpoints $\mathcal{Y}_t$ is a convex and bounded subset of $\mathcal{Y}_{\pi,s}$.
		\item The offset cost function $V_O: \mathbb{R}^p \to \mathbb{R}$ is a convex positive definite function such that the minimizer
		\begin{equation}
			\label{eq:ys_choice}
			y_s^*=\arg \min_{y_s \in \mathcal{Y}_t} V_O(y_s-y_t)
		\end{equation}
		is unique. Moreover, there exists a $\mathcal{K}_\infty$ function $\alpha_O(\cdot)$ such that
		\begin{equation}
			V_O(y_s-y_t) - V_O(y_s^* - y_t) \geq \alpha_O (\Vert y_s - y_s^* \Vert) \; \forall y_s \in \mathcal{Y}_t.
		\end{equation}
	\end{enumerate}
\end{assumption}
\begin{remark}
	Assumption \ref{ass:rnmpct_offset_function} requires the boundedness of the set $\mathcal{Y}_t$. This can be obtained by either having an output function $h(\cdot,\cdot)$ that is bounded in $\mathcal{Z}$, or by accepting to let the MPC steer the closed-loop output to setpoints belonging to a bounded set $\mathcal{Y}_t$ even when $\mathcal{Y}_s$ has no bounds. In case of a compact constraint set $\mathcal{Z}$, $\mathcal{Y}_s$ is bounded, therefore allowing the designer to choose $\mathcal{Y}_t = \mathcal{Y}_s$ if $\mathcal{Y}_s$ is also convex.
	
	Moreover, given that the convexity of $\mathcal{Y}_t$ implies that $\mathcal{Y}_t$ is also connected, in case of systems whose output equilibria belong to disjoint regions, it will be necessary to set $\mathcal{Y}_t$ as a convex subset of one of those regions, thus impeding to stabilize all the equilibria of the system.
\end{remark}
Assumption \ref{ass:rnmpct_offset_function} implies that, given the boundedness of $\mathcal{Y}_t$, there exists $V_{O,\mathrm{max}} > 0$ such that $V_O( y_s - \check{y}_s) \leq V_{O,\mathrm{max}}$ for any $y_s,\check{y}_s \in \mathcal{Y}_t$.

For reasons that will become more clear later in the text, in order to prove the input-to-state stability of the closed-loop system, it is necessary to make a further (non-restrictive) assumption on the function $\alpha_O(\cdot)$. For the sake of simplicity, the next assumption makes use of the set $\mathcal{S}$, which is defined as
\begin{equation}
	\mathcal{S} \coloneqq \left\{ \Vert y_s - \check{y}_s \Vert: y_s,\check{y}_s \in \mathcal{Y}_t \right\}.
\end{equation}
Note that, due to the boundedness of $\mathcal{Y}_t$, $\mathcal{S} \subset \mathbb{R}$ is a bounded interval of real numbers, in detail $\mathcal{S} \equiv [0, \sup \Vert y_s - \check{y}_s \Vert]$, with $y_s,\check{y}_s \in \mathcal{Y}_t$. 
\begin{assumption}[Joint conditions on $V_f(\cdot,\cdot)$ and $V_O(\cdot)$]
	\label{ass:rnmpct_alpha_O_alpha_f_ratio}
	Given the Lipschitz constant $L_g$, for which $\Vert x_s - \check{x}_s \Vert \leq L_g \Vert y_s - \check{y}_s \Vert$, the functions $V_f(\cdot,\cdot)$ and $V_O(\cdot)$ are such that
	\begin{equation}
		\label{eq:l1}
		\begin{split}
			\mathfrak{b}_1 &\coloneqq \mathrm{inf} \Bigg\{ \frac{\alpha_O(s)}{4 \bar{\alpha}_f L_g^2 s^2}: \; s \in \mathcal{S} \setminus \{ 0 \} \Bigg\} > 0,
		\end{split}
	\end{equation}
	\begin{equation}
		\label{eq:l2}
		\begin{split}
			\mathfrak{b}_2 &\coloneqq \mathrm{inf} \Bigg\{ \frac{\alpha_O(s_2)   -   \alpha_O(s_1)}{4 \bar{\alpha}_f L_g^2 s_2^2   -   4 \bar{\alpha}_f L_g^2 s_1^2}:   s_1,s_2 \in \mathcal{S}, s_2 > s_1   \Bigg\} > 0.
		\end{split}
	\end{equation}
\end{assumption}
\begin{remark}
	Thanks to the boundedness of $\mathcal{S}$, Assumption \ref{ass:rnmpct_alpha_O_alpha_f_ratio} is not particularly restrictive, as it can easily be verified by choosing an offset function $V_O(\cdot)$ such that, for $s \to 0^+$, $\alpha_O(s)$ goes to 0 at most as rapidly as $s^2$. 
	Examples of functions that verify the last conditions are quadratic functions and polynomials of degree equal to or lower than 2.
\end{remark}
In the following, we analyze how assumptions \ref{ass:rnmpct_continuity}-\ref{ass:rnmpct_alpha_O_alpha_f_ratio} are sufficient to ensure the recursive feasibility and the stabilizing capabilities of our robust MPC for tracking for perturbed nonlinear systems. Given that the proof of stability of the closed loop is complex, we first address recursive feasibility in Lemma \ref{lem:rnmpct_robust_feasibility}.
\begin{lemma}[Recursive feasibility of the robust MPC for tracking for nonlinear systems]
	\label{lem:rnmpct_robust_feasibility}
	Consider the system $x(k+1) = f \big( x(k), \pi ( y_s, \allowbreak x(k),v(k) ),w(k) \big)$ and the set sequence $\big\{\mathcal{Z}_\pi(j)\big\}_{j \geq 0}$, based on $\big\{\mathcal{R}(j)\big\}_{j \geq 0}$ and $\big\{\mathcal{F}(j)\big\}_{j \geq 0}$, which satisfy Assumptions \ref{ass:rnmpct_F_sequence} and \ref{ass:rnmpct_R_sequence}. Let the triplet $\big( v_f(y_s^0),\Omega(y_s^0),\mathcal{T} \big)$ fulfill Assumption \ref{ass:rnmpct_robust_terminal_set}. Consider now a feasible state $x(k) \in \mathcal{X}_{N_p}$ and the optimal solution $\big(\mathbf{v}^0(k),y_s^0(k)\big)$ to the problem $P_{N_p}(x(k),y_t)$. Let $x(k+1)$ be the next uncertain state and define a candidate artificial reference $y_s^+(k+1) = y_s^0(k)$ and a sequence of inputs $\mathbf{v}^+(k+1) \coloneqq \{v^+(0 \vert k+1), \hdots,v^+(N_p-1 \vert k+1)\} = \{ v^0(1 \vert k),\hdots,v^0(N_p-1 \vert k),v_f(y_s^+(k+1)) \}$. Then the following properties hold:
	\begin{enumerate}
		\item $\big( \phi_\pi(j;y_s^+(k+1),x(k+1),\mathbf{v}^+(k+1),\mathbf{0}),v^+(j \vert k+1) \big) \in \mathcal{Z}_\pi(j)$, $\forall j \in [0,N_p-1]$,
		\item $\big( \phi_\pi(N_p;y_s^+(k+1),x(k+1),\mathbf{v}^+(k+1),\mathbf{0}),y_s^+(k+1) \big) \in \mathcal{T}$.
	\end{enumerate} 
\end{lemma}
\begin{proof}
	The proof is provided in the Appendix \ref{sec:app_proofs}.
\end{proof}
The following theorem ensures the input-to-state stability of the closed-loop system controlled by the proposed robust MPC for tracking and convergence of the system output to a neighborhood of the best admissible steady output.
\begin{theorem}[Stability of the tube-based MPC for tracking for perturbed nonlinear systems]
	\label{thm:rnmpct_iss_convergence}
	Suppose that Assumptions \ref{ass:rnmpct_continuity}-\ref{ass:rnmpct_alpha_O_alpha_f_ratio} hold, and consider a given constant setpoint $y_t$. Then for any feasible initial state $x \in \mathcal{X}_{N_p}$, the system controlled by $\kappa_{N_p}(x,y_t)$ fulfills the constraints throughout the time, the equilibrium $x_s^* = g_x(y_s^*)$ is input-to-state stable for the closed-loop system, whose output converges to a bounded set containing the best admissible steady output $y_s^*$.
\end{theorem}
\begin{proof}
	The proof is provided in the Appendix \ref{sec:app_proofs}.
\end{proof}

\section{Case study}
\label{sec:case_study}
In this section, the here-proposed controller is tested on a simulated four-tank system with uncertain parameters. In detail, considering a sampling time $T_s = 15$s, the system under control is the fourth-order Runge Kutta discretization of the continuous-time system described by the equations
\begin{subequations}
	\begin{align}
		\frac{dh_1(t)}{dt} &= -\frac{a_1}{S}\sqrt{2gh_1(t)} + \frac{a_3}{S}\sqrt{2gh_3(t)} + \frac{\gamma_a}{3600S}q_a(t),\\
		\frac{dh_2(t)}{dt} &= -\frac{a_2}{S}\sqrt{2gh_2(t)} + \frac{a_4}{S}\sqrt{2gh_4(t)} + \frac{\gamma_b}{3600S}q_b(t),\\
		\frac{dh_3(t)}{dt} &= -\frac{a_3}{S}\sqrt{2gh_3(t)} + \frac{(1 - \gamma_b)}{3600S}q_b(t),\\
		\frac{dh_4(t)}{dt} &= -\frac{a_4}{S}\sqrt{2gh_4(t)} + \frac{(1 - \gamma_a)}{3600S}q_a(t),
	\end{align}
\end{subequations}
where:
\begin{itemize}
	\item $h_1,h_2,h_3,h_4 \; \mathrm{[m]}$ are the water levels in the four tanks, which constitute the state of the system;
	\item $q_a,q_b \; \mathrm{[m^3/h]}$ are the flows of the two pumps of the system, and are the system inputs;
	\item $\gamma_a,\gamma_b \; [/]$ are the parameters that define how the two three-way valves operate;
	\item $a_1,a_2,a_3,a_4 \; \mathrm{[m^2]}$ are the sections of the tank holes, which determine the discharge speed of the four tanks;
	\item $S \; \mathrm{[m^2]}$ is the cross-section of all the tanks;
	\item $g = 9.81 \; \mathrm{[m/s^2]}$ is the standard acceleration of gravity.
\end{itemize} 
The output of the system is $y(t) = (h_1(t),h_2(t))$. In the context of our work, we consider the parameters $\gamma_a,\gamma_b$ to be uncertain, namely
\begin{equation}
	\gamma_a = \bar{\gamma}_a + w_1, \quad
	\gamma_b = \bar{\gamma}_b + w_2,
\end{equation}
where $\bar{\gamma}_a$ and $\bar{\gamma}_b$ are the nominal values of the parameters $\gamma_a$ and $\gamma_b$, while $w_1$ and $w_2$ describe the uncertainties of such parameters and are supposed to be bounded. In detail,
\begin{equation}
	\label{eq:rnmpct_4t_w}
	\vert w_1 \vert \leq \bar{w}_1, \quad
	\vert w_2 \vert \leq \bar{w}_2.
\end{equation}
The system inputs and states have to verify the following constraints:
\begin{subequations}
	\label{eq:rnmpct_4t_constraints}
	\begin{align}
		h_{i,\mathrm{min}} &\leq h_i \leq h_{i,\mathrm{max}}, \; i \in [1,4],\\
		q_{j,\mathrm{min}} &\leq q_j \leq q_{j,\mathrm{max}}, \; j \in \{a,b\}.
	\end{align}
\end{subequations}
The parameters of the simulated plant are taken from \cite{Limon2018_NMPCT} while the bounds of the uncertainty variables $w_1$ and $w_2$ are chosen as $\bar{w}_1 = \bar{w}_2 = 0.005$.
The control goal is to track a piece-wise constant reference signal, starting from the initial condition $(h_1(0),h_2(0),h_3(0),h_4(0)) = (0.2837,0.2943,0.2168,0.2864)\mathrm{m}$. 

In the following, we revert to a more generic notation and define the states and inputs of the discretized system as $x(k) = (x_1(k),x_2(k),x_3(k),x_4(k)) \coloneqq (h_1(k\cdot T_s),h_2(k\cdot T_s),h_3(k\cdot T_s),h_4(k\cdot T_s))$ and $u(k) = (u_1(k),u_2(k)) \coloneqq (q_a(k\cdot T_s),q_a(k\cdot T_s))$. Moreover, we define the sets $\mathcal{X}$, $\mathcal{U}$ and $\mathcal{W}$ as sets such that $x(k) \in \mathcal{X}$, $u(k) \in \mathcal{U}$ and $w(k) \coloneqq (w_1(k\cdot T_s),w_2(k\cdot T_s)) \in \mathcal{W}$ imply the verification of the constraints \eqref{eq:rnmpct_4t_w} and \eqref{eq:rnmpct_4t_constraints}. To reach the previously stated control goal, we employ a control law 
\begin{equation}
	\pi \left( y_s^0(k),x(k),v^0(0 \vert k) \right) = K (x(k) - g_x(y_s^0(k))) + v^0(0 \vert k),
\end{equation}
where $y_s^0(k)$ and $v^0(0 \vert k)$ are obtained by solving the optimal control problem $P_{N_p}(x(k),y_t)$, with $N_p = 4$, employing a quadratic stage cost function $\ell(x - xs,v - v_s) = \big( x - x_s \big)^\top Q \big( x - x_s \big) + \big( v - v_s \big)^\top R \big( v - v_s \big)$, a quadratic terminal cost function $V_f(x - x_s,y_s) = (x - x_s)^\top P (x - x_s)$, and a quadratic offset cost function $V_O(y_s - y_t) = (y_s - y_t)^\top T (y_s - y_t)$. In detail, $Q = \mathrm{diag}(5,2.5,1,1)$, $R = 0.01I_2$ and $T = 10^4I_2$. 

The matrices $K$ and $P$ are obtained using the method proposed in the Appendix~\ref{alg:terminal_ingredients}. In detail, the nonlinear system under control is locally modelled as an LTV system, and matrices $K$ and $P$ verifying the Lyapunov inequality \eqref{eq:rnmpct_lyap_ineq} can be computed by resorting to LMIs. The obtained matrices are:
\begin{equation}
	K = \begin{bmatrix}
		-0.2117	& -1.0663 & 1.0632 & -2.2400\\
		-2.9453 & 0.3358 & -2.9568 & 1.5166
	\end{bmatrix},
\end{equation}
\begin{equation}
	P = \begin{bmatrix}
		23.7301 & -4.6424 & 10.2060 & -8.5938\\
		-4.6425 & 8.8324 & -4.7139 & 3.8709\\
		10.2060 & -4.7139 & 12.6887 & -7.8934\\
		-8.5938 & 3.8709 & -7.8934 & 9.4075
	\end{bmatrix}.
\end{equation}
Given that $\mathcal{X}$, $\mathcal{U}$ and $\mathcal{W}$ are compact and that the system under control is Lipschitz in $\mathcal{X} \times \mathcal{U} \times \mathcal{W}$, the system equations are also component-wise Lipschitz continuous in such set. Therefore, we first compute the component-wise Lipschitz continuity constants of the discretized system and then compute the sequences $\left\{ \mathcal{F}(j) \right\}_{\geq 0}$ and $\left\{ \mathcal{R}(j) \right\}_{\geq 0}$ following the algorithms presented in the Appendices \ref{alg:lipschitz} and \ref{alg:F_R}. The found component-wise Lipschitz constants are reported in Tables \ref{tab:rnmpct_Lx}, \ref{tab:rnmpct_Lu} and \ref{tab:rnmpct_Lw}.
\begin{table}
	\centering
	\caption{Values of the constants $L_{x,ia}$ of the four-tank system.}
	\label{tab:rnmpct_Lx}
	\arrayrulecolor{black}
	\begin{tabular}{|c|c|c|c|c|l|} 
		\hline
		\rowcolor{myblue!10} \multicolumn{2}{|c|}{{\cellcolor{myblue!10}}}                               & \multicolumn{4}{c|}{$a$}                                                                                                   \\ 
		\hhline{|>{\arrayrulecolor{myblue!10}}-->{\arrayrulecolor{black}}----|}
		\rowcolor{myblue!10} \multicolumn{2}{|c|}{\multirow{-2}{*}{{\cellcolor{myblue!10}}$L_{x,ia}$}}   & \textit{1}                  & \textit{2}                  & \textit{3}                  & \multicolumn{1}{c|}{\textit{4}}  \\ 
		\hline
		{\cellcolor{myblue!10}}                      & {\cellcolor{myblue!10}}\textit{1}                 & 0.9388                      & 0.0250                      & 0.1375                      & 0.0500                           \\ 
		\hhline{|>{\arrayrulecolor{myblue!10}}->{\arrayrulecolor{black}}-----|}
		{\cellcolor{myblue!10}}                      & {\cellcolor{myblue!10}}\textit{2}                 & 0.0850                      & 0.9388                      & 0.0795                      & 0.1600                           \\ 
		\hhline{|>{\arrayrulecolor{myblue!10}}->{\arrayrulecolor{black}}-----|}
		{\cellcolor{myblue!10}}                      & {\cellcolor{myblue!10}}\textit{3}                 & 0.1225                      & 0.0145                      & 0.8375                      & 0.0750                           \\ 
		\hhline{|>{\arrayrulecolor{myblue!10}}->{\arrayrulecolor{black}}-----|}
		\multirow{-4}{*}{{\cellcolor{myblue!10}}$i$} & {\cellcolor{myblue!10}}\textit{4}                 & \multicolumn{1}{l|}{0.0125} & \multicolumn{1}{l|}{0.0600} & \multicolumn{1}{l|}{0.0600} & 0.8500                           \\
		\hline
	\end{tabular}
\end{table}

\begin{table}
	\centering
	\caption{Values of the constants $L_{v,ib}$ of the four-tank system.}
	\label{tab:rnmpct_Lu}
	\arrayrulecolor{black}
	\begin{tabular}{|c|c|c|c|} 
		\hline
		\rowcolor{myblue!10} \multicolumn{2}{|c|}{{\cellcolor{myblue!10}}}                               & \multicolumn{2}{c|}{$b$}                                   \\ 
		\hhline{|>{\arrayrulecolor{myblue!10}}-->{\arrayrulecolor{black}}--|}
		\rowcolor{myblue!10} \multicolumn{2}{|c|}{\multirow{-2}{*}{{\cellcolor{myblue!10}}$L_{v,ib}$}}   & \textit{1}                  & \textit{2}                   \\ 
		\hline
		{\cellcolor{myblue!10}}                      & {\cellcolor{myblue!10}}\textit{1}                 & 0.0225                      & 0.0250                       \\ 
		\hhline{|>{\arrayrulecolor{myblue!10}}->{\arrayrulecolor{black}}---|}
		{\cellcolor{myblue!10}}                      & {\cellcolor{myblue!10}}\textit{2}                 & 0.0500                      & 0.0350                       \\ 
		\hhline{|>{\arrayrulecolor{myblue!10}}->{\arrayrulecolor{black}}---|}
		{\cellcolor{myblue!10}}                      & {\cellcolor{myblue!10}}\textit{3}                 & 0.0100                      & 0.1700                       \\ 
		\hhline{|>{\arrayrulecolor{myblue!10}}->{\arrayrulecolor{black}}---|}
		\multirow{-4}{*}{{\cellcolor{myblue!10}}$i$} & {\cellcolor{myblue!10}}\textit{4}                 & \multicolumn{1}{l|}{0.1500} & \multicolumn{1}{l|}{0.0100}  \\
		\hline
	\end{tabular}
\end{table}

\begin{table}
	\centering
	\caption{Values of the constants $L_{w,ic}$ of the four-tank system.}
	\label{tab:rnmpct_Lw}
	\arrayrulecolor{black}
	\begin{tabular}{|c|c|c|c|} 
		\hline
		\rowcolor{myblue!10} \multicolumn{2}{|c|}{{\cellcolor{myblue!10}}}                               & \multicolumn{2}{c|}{$c$}                                   \\ 
		\hhline{|>{\arrayrulecolor{myblue!10}}-->{\arrayrulecolor{black}}--|}
		\rowcolor{myblue!10} \multicolumn{2}{|c|}{\multirow{-2}{*}{{\cellcolor{myblue!10}}$L_{w,ic}$}}   & \textit{1}                  & \textit{2}                   \\ 
		\hline
		{\cellcolor{myblue!10}}                      & {\cellcolor{myblue!10}}\textit{1}                 & 0.2400                      & 0.0125                       \\ 
		\hhline{|>{\arrayrulecolor{myblue!10}}->{\arrayrulecolor{black}}---|}
		{\cellcolor{myblue!10}}                      & {\cellcolor{myblue!10}}\textit{2}                 & 0.0088                      & 0.2638                       \\ 
		\hhline{|>{\arrayrulecolor{myblue!10}}->{\arrayrulecolor{black}}---|}
		{\cellcolor{myblue!10}}                      & {\cellcolor{myblue!10}}\textit{3}                 & 0.00006                     & 0.2675                       \\ 
		\hhline{|>{\arrayrulecolor{myblue!10}}->{\arrayrulecolor{black}}---|}
		\multirow{-4}{*}{{\cellcolor{myblue!10}}$i$} & {\cellcolor{myblue!10}}\textit{4}                 & \multicolumn{1}{l|}{0.2450} & \multicolumn{1}{l|}{0.0125}  \\
		\hline
	\end{tabular}
\end{table}

Furthermore, as the computed sets $\mathcal{F}(j)$ and $\mathcal{R}(j)$ are box-shaped, i.e.
\begin{equation}
	\mathcal{F}(j) = \{ x \in \mathbb{R}^4: \; \vert x_i \vert \leq \bar{x}_{i,\mathcal{F}(j)}, \forall i \in [1,4] \},
\end{equation}
\begin{equation}
	\mathcal{R}(j) = \{ x \in \mathbb{R}^4: \; \vert x_i \vert \leq \bar{x}_{i,\mathcal{R}(j)}, \forall i \in [1,4] \},
\end{equation}
we provide the computed $\bar{x}_{i,\mathcal{F}(j)}$ and $\bar{x}_{i,\mathcal{R}(j)}$ in Table \ref{tab:rnmpct_F_R_values}.
\begin{table}
	\centering
	\caption{Description of the sequences $\{ \mathcal{F}(j) \}_{j \geq 0}$ and $\{ \mathcal{R}(j) \}_{j \geq 0}$.}
	\label{tab:rnmpct_F_R_values}
	\arrayrulecolor{black}
	\resizebox{\columnwidth}{!}{%
		\begin{tabular}{|c|c|c|c|c|c|c|} 
			\hline
			\rowcolor{myblue!10} \multicolumn{2}{|c|}{$j$}                                                                                                  & \textit{0} & \textit{1}                  & \textit{2}                  & \textit{3}                  & \textit{4}                   \\ 
			\hline
			{\cellcolor{myblue!10}}                                   & {\cellcolor{myblue!10}}$\bar{x}_1,\mathcal{F}(j)$                      & 0.0019     & 0.0022                      & 0.0025                      & 0.0028                      & 0.0032                       \\ 
			\hhline{|>{\arrayrulecolor{myblue!10}}->{\arrayrulecolor{black}}------|}
			{\cellcolor{myblue!10}}                                   & {\cellcolor{myblue!10}}$\bar{x}_2,\mathcal{F}(j)$                      & 0.0020     & 0.0025                      & 0.0031                      & 0.0036                      & 0.0041                       \\ 
			\hhline{|>{\arrayrulecolor{myblue!10}}->{\arrayrulecolor{black}}------|}
			{\cellcolor{myblue!10}}                                   & {\cellcolor{myblue!10}}$\bar{x}_3,\mathcal{F}(j)$                      & 0.0020     & 0.0021                      & 0.0022                      & 0.0023                      & 0.0025                       \\ 
			\hhline{|>{\arrayrulecolor{myblue!10}}->{\arrayrulecolor{black}}------|}
			\multirow{-4}{*}{{\cellcolor{myblue!10}}$\mathcal{F}(j)$} & \multicolumn{1}{l|}{{\cellcolor{myblue!10}}$\bar{x}_4,\mathcal{F}(j)$} & 0.0019     & \multicolumn{1}{l|}{0.0019} & \multicolumn{1}{l|}{0.0020} & \multicolumn{1}{l|}{0.0020} & \multicolumn{1}{l|}{0.0021}  \\ 
			\hline
			{\cellcolor{myblue!10}}                                   & {\cellcolor{myblue!10}}$\bar{x}_1,\mathcal{R}(j)$                      & 0          & 0.0019                      & 0.0041                      & 0.0066                      & 0.0094                       \\ 
			\hhline{|>{\arrayrulecolor{myblue!10}}->{\arrayrulecolor{black}}------|}
			{\cellcolor{myblue!10}}                                   & {\cellcolor{myblue!10}}$\bar{x}_2,\mathcal{R}(j)$                      & 0          & 0.0020                      & 0.0046                      & 0.0076                      & 0.0112                       \\ 
			\hhline{|>{\arrayrulecolor{myblue!10}}->{\arrayrulecolor{black}}------|}
			{\cellcolor{myblue!10}}                                   & {\cellcolor{myblue!10}}$\bar{x}_3,\mathcal{R}(j)$                      & 0          & 0.0020                      & 0.0041                      & 0.0063                      & 0.0086                       \\ 
			\hhline{|>{\arrayrulecolor{myblue!10}}->{\arrayrulecolor{black}}------|}
			\multirow{-4}{*}{{\cellcolor{myblue!10}}$\mathcal{R}(j)$} & \multicolumn{1}{l|}{{\cellcolor{myblue!10}}$\bar{x}_4,\mathcal{R}(j)$} & 0          & \multicolumn{1}{l|}{0.0019} & \multicolumn{1}{l|}{0.0038} & \multicolumn{1}{l|}{0.0058} & \multicolumn{1}{l|}{0.0078}  \\
			\hline
		\end{tabular}%
	}
\end{table}

The chosen bounded subset $\mathcal{Y}_t$ of feasible setpoints with inactive constraints is provided in Figure \ref{fig:rnmpct_Ys}, while a robust invariant set for tracking $\mathcal{T}$ verifying all the conditions of Assumption \ref{ass:rnmpct_robust_terminal_set} is found as
\begin{equation}
	\mathcal{T} = \left\{ x: V_f(x-x_s,y_s) \leq \rho \right\} \times \mathcal{Y}_t,
\end{equation}
where $\rho = 0.1273$.

\begin{figure}[h]
	\begin{center}
		\caption{$\mathcal{Y}_t$ and projection of $\mathcal{Z}_\pi(4)$ on the first two state dimensions.} 
		\label{fig:rnmpct_Ys}
		\includegraphics[width=0.95\columnwidth,trim = 3.4cm 9cm 4.1cm 9.25cm,clip]{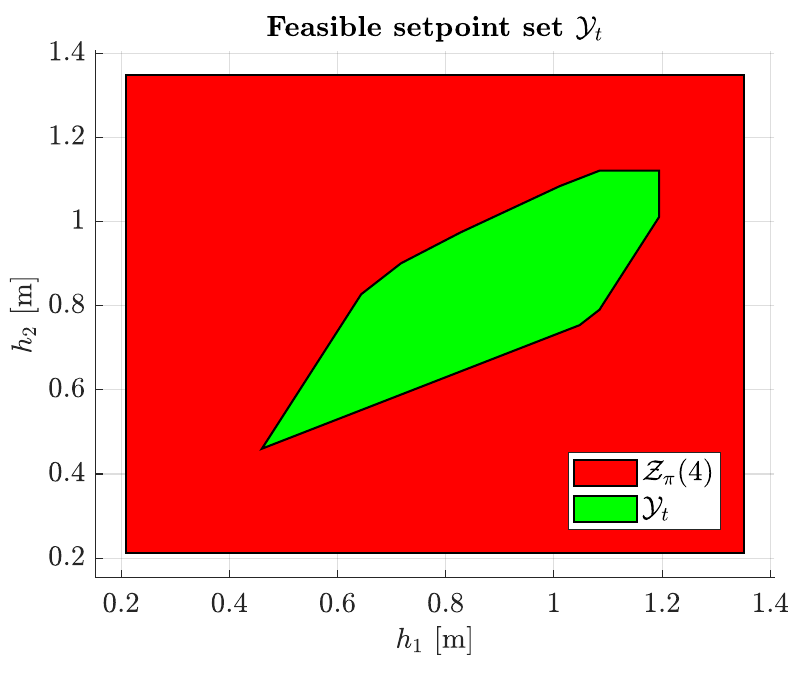}    
	\end{center}
\end{figure}

We run 100 simulations with different values for $w_1$ and $w_2$. In detail, we test all the combinations that are obtained by using values for $w_1$ and $w_2$ in the set $\{ -0.00500, -0.00390, -0.00280, -0.00170, -0.00055, \allowbreak 0.00055, 0.0017, 0.00280, 0.00390, 0.00500 \}$, so as to cover the entire set $\mathcal{W}$. The simulation time is always equal to 100min and the reference $y_t$ is considered to take values $(0.65,0.65)$, $(0.35,0.35)$, $(0.60,0.75)$ and $(0.90,0.75)$.

The obtained state and input trajectories are shown in Figures \ref{fig:rnmpct_state_trajectories} and \ref{fig:rnmpct_input_trajectories} respectively, and prove how the here-proposed robust MPC for tracking is able to ensure input-to-state stability and constraint satisfaction of the closed-loop system, and to steer the system output to a neighborhood of the best admissible steady output, even in case of abruptly changing and infeasible setpoints.

\begin{figure}[h]
	\begin{center}
		\caption{State trajectories obtained in 100 different scenarios. The constraints are always verified, and the output of the system (in this case the first two states) converges to a neighborhood of the best admissible steady output.} 
		\label{fig:rnmpct_state_trajectories}
		\includegraphics[width=\columnwidth,trim = 3.5cm 6cm 3.5cm 6cm,clip]{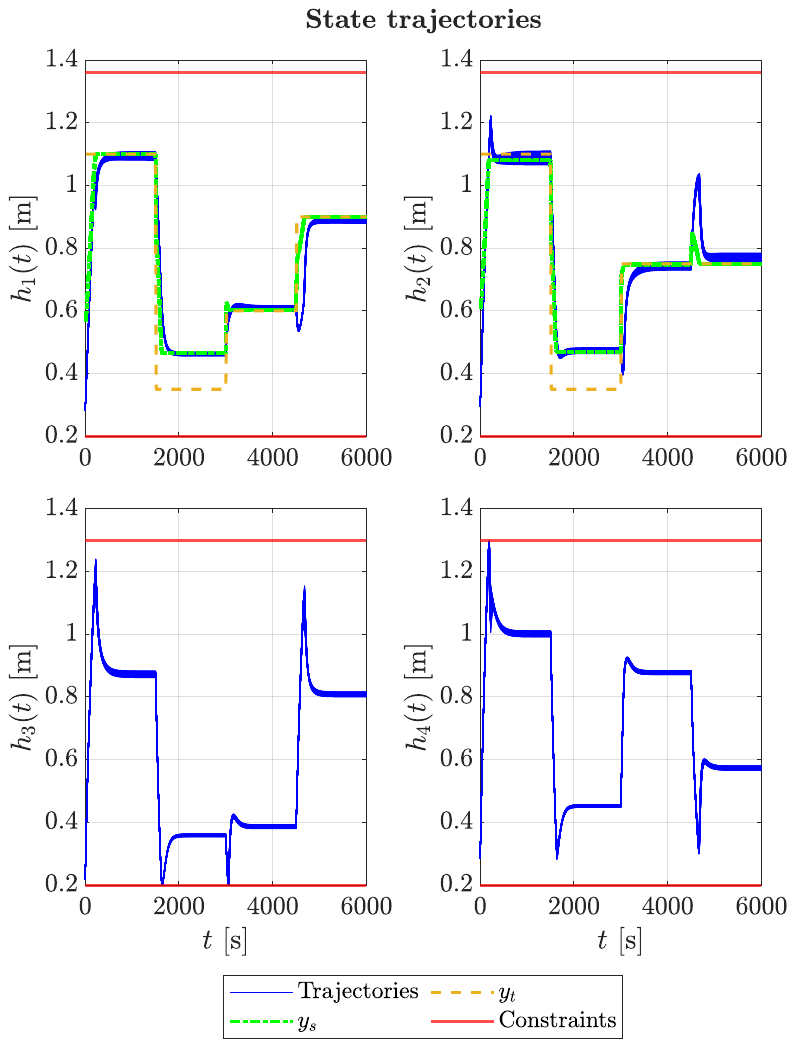}    
	\end{center}
\end{figure}

\begin{figure}[h]
	\begin{center}
		\caption{Input trajectories obtained in 100 different scenarios. As expected, also the input values always verify the constraints.} 
		\label{fig:rnmpct_input_trajectories}
		\includegraphics[width=\columnwidth,trim = 3.75cm 10.5cm 3.5cm 10.25cm,clip]{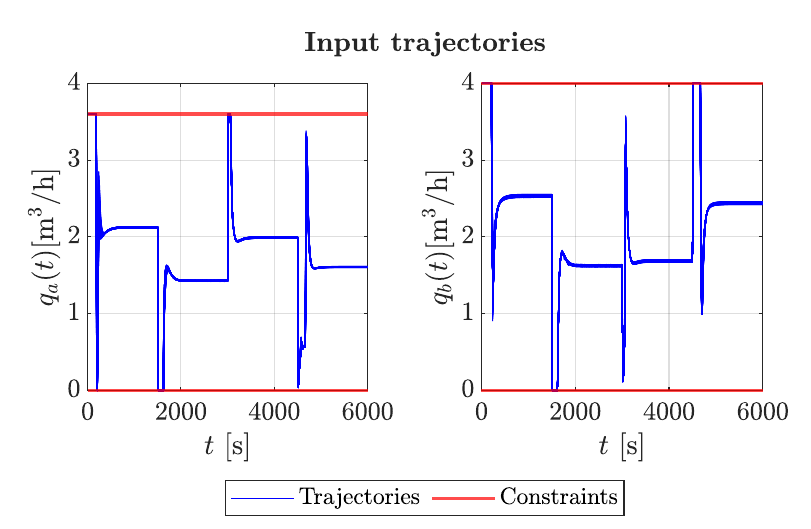}    
	\end{center}
\end{figure}
\section{Conclusion}
\label{sec:conclusion}
A novel robust nonlinear MPC for tracking piece-wise constant reference signals for constrained nonlinear systems subject to bounded perturbations has been proposed. Under mild assumptions on the system model and the components of the chosen cost function, the here-proposed controller is proved to be recursively feasible, to make the closed-loop system input-to-state stable, and to ensure the convergence of the system output to a bounded set containing the best admissible steady output. A case study about the control of a four-tank system with uncertain parameters proved the deployability of our controller and gave concrete evidence of the feasibility and stability properties ensured by the theory.

\appendix
\section{Proofs}   
\label{sec:app_proofs}
For the following proofs, consider a feasible state $x(k) \in \mathcal{X}_{N_p}$ and the optimal solution $\big(\mathbf{v}^0(x(k),y_t),y_s^0(x(k),\allowbreak y_t)\big)$ to the problem $P_{N_p}(x(k),y_t)$, let $x(k+1)$ be the uncertain successor state and define a candidate solution $\big(\mathbf{v}^+(k+1),y_s^+(k+1)\big)$ to the problem $P_{N_p}(x(k+1),y_t)$ such that $y_s^+(k+1) = y_s^0(k)$ and $\mathbf{v}^+(k+1) = \{ v^0(1 \vert k),\hdots,v^0(N_p - 1 \vert k),v_f(y_s^+(k+1)) \}$.
\subsection{Proof of Lemma \ref{lem:rnmpct_robust_feasibility}}
\textbf{First statement:} Since $\phi_\pi(j;y_s^+(k+1),x(k+1),\allowbreak \mathbf{v}^+(k+1),\mathbf{0}) - \phi_\pi(j+1;y_s^0(k), \allowbreak x(k),\mathbf{v}^0(k),\mathbf{0}) \in \mathcal{F}(j)$, assumption \ref{ass:rnmpct_F_sequence} allows us to state that

\begin{equation}
	\begin{split}
		&\phi_\pi(j;y_s^+(k+1),x(k   + 1),\mathbf{v}^+(k+1),\mathbf{0})\\ &\in \phi_\pi(j   +   1;y_s^0(k),x(k),\mathbf{v}^0(k),\mathbf{0}) \oplus \mathcal{F}(j) \\
		&\subseteq \phi_\pi(j+1;y_s^0(k),x(k),\mathbf{v}^0(k),\mathbf{0}) \oplus \mathcal{R}(j+1) \ominus \mathcal{R}(j).
	\end{split}
\end{equation}

Since $\big(\phi_\pi(j + 1;y_s^0(k),x(k),\mathbf{v}^0(k),\mathbf{0}),v^0(j+1 \vert k)\big)   \in   \mathcal{Z}_\pi(j + 1)$, we have that

\begin{equation}
	\begin{split}
		&\big(\phi_\pi(j;y_s^+(k+1),x(k+1),\mathbf{v}^+(k+1),\mathbf{0}),v^0(j+1 \vert k)\big)\\ &\in \Big( \mathcal{Z}_\pi(j+1) \oplus \big( \mathcal{R}(j+1) \times \{ 0 \} \big) \ominus \big( \mathcal{R}(j) \times \{ 0 \} \big) \Big).
	\end{split}
\end{equation}

By definition,

\begin{equation}
	\begin{split}
		\mathcal{Z}_\pi(j+1) \oplus \big( \mathcal{R}(j+1) \times \{ 0 \} \big) \subseteq \mathcal{Z}_\pi
	\end{split}
\end{equation}

and $v^0(j+1 \vert k)=v^+(j \vert k)$, which implies that

\begin{equation}
	\begin{split}
		&\big(\phi_\pi(j;y_s^+(k+1),x(k+1),\mathbf{v}^+(k+1),\mathbf{0}),v^+(j \vert k+1)\big)\\ &\in \mathcal{Z}_\pi \ominus \big( \mathcal{R}(j) \times \{ 0 \} \big)   \subseteq   \mathcal{Z}_\pi(j).
	\end{split}
\end{equation}

\textbf{Second statement:} The feasibility of the solution $\big(\mathbf{v}^0(k),y_s^0(k)\big)$ implies that $\big( 
\phi_\pi(N_p; \allowbreak y_s^0(k),x(k),\mathbf{v}^0(k), \mathbf{0}),y_s^0(k) \big) \in \mathcal{T}$; in particular, $\phi_\pi(N_p; \allowbreak y_s^0(k),x(k),\mathbf{v}^0(k),  \mathbf{0}) \in \mathcal{X}_f(y_s^0(k))$ and $y_s^0(k) \in \mathcal{Y}_{\pi,s}$. On the other hand, from the definition of the set $\mathcal{F}(N_p-1)$, we have that

\begin{equation}
	\begin{split}
		&\phi_\pi(N_p-1;y_s^+(k+1),x(k+1),\mathbf{v}^+(k+1),\mathbf{0})\\ &\in \phi_\pi(N_p;y_s^0(k),x(k),\mathbf{v}^0(k),\mathbf{0}) \oplus \mathcal{F}(N_p-1).
	\end{split}
\end{equation}

Then

\begin{equation}
	\begin{split}
		&\phi_\pi(N_p-1;y_s^+(k+1),x(k+1),\mathbf{v}^+(k+1),\mathbf{0})\\ &\in \mathcal{X}_f(y_s^+(k+1)) \oplus \mathcal{F}(N_p-1).
	\end{split}
\end{equation}

Since $y_s^+(k+1) = y_s^0(k)$, $\mathcal{X}_f(y_s^+(k+1)) = \mathcal{X}_f(y_s^0(k))$ and $y_s^+(k+1) \in \mathcal{Y}_{\pi,s}$, that is $\big( \phi_\pi(N_p;y_s^+(k+1),x(k+1),\mathbf{v}^+(k+1),\mathbf{0}), \allowbreak y_s^+(k+1) \big) \in \mathcal{T}$.

\subsection{Proof of Theorem \ref{thm:rnmpct_iss_convergence}}
This proof is divided in two parts. First, the recursive feasibility of $P_{N_p}(\cdot,\cdot)$ is proved. In the second part, the input-to-state stability of $x_s^*$ and the convergence to a neighborhood of the best admissible steady output $y_s^*$ are proved. 

\textbf{Recursive feasibility:} the proof of recursive feasibility is given in Lemma \ref{lem:rnmpct_robust_feasibility}.

\textbf{Input-to-state stability and convergence:} Let $(x_s^*,u_s^*)$ be the equilibrium point given by $x_s^* = g_x(y_s^*)$ and $u_s^* = g_u(y_s^*)$. We want to prove that the closed-loop system is ISS and converges to a neighborhood of the best feasible steady output $y_s^*$. To this end, we show that the function 

\begin{equation} \label{eq:W_function}
	W(x(k),y_t) \coloneqq V_{N_p}^0(x(k),y_t) - V_O(y_s^* - y_t)
\end{equation}

is an ISS-Lyapunov function. Therefore, we need to prove that:
\begin{itemize}
	\item There exist two $\mathcal{K}_\infty$ functions $\bar{\alpha}(\cdot)$ and $\underline{\alpha}(\cdot)$ such that $\underline{\alpha}(\Vert x(k) - x_s^* \Vert) \leq W(x(k),y_t) \leq \bar{\alpha}(\Vert x(k) - x_s^* \Vert)$;
	\item There exist a $\mathcal{K}_\infty$ function $\alpha(\cdot)$ and a $\mathcal{K}$ function $\lambda(\cdot)$ such that
	 \begin{equation}
		\begin{split}
			&W(x(k+1),y_t) - W(x(k),y_t)\\ &\leq \lambda(\Vert w(k) \Vert) - \alpha (\Vert x(k) - x_s^* \Vert).
		\end{split}
	\end{equation} 
\end{itemize}
We start by looking for a $\mathcal{K}_\infty$ function that bounds $W(x(k),y_t)$ from below. From the definition of the function $W(\cdot,y_t)$, and due to Assumptions \ref{ass:rnmpct_tracking_bounds} and \ref{ass:rnmpct_offset_function}, we have that
 \begin{equation}
	\begin{split}
		W(x(k),y_t) &\geq \ell \big( x(k) - x_s^0(k), v^0(0 \vert k) - v_s^0(k) \big)\\ &+ V_O(y_s^0(k) - y_t) - V_O(y_s^* - y_t)\\
		&\geq \alpha_\ell ( \Vert x(k) - x_s^0(k) \Vert ) + \alpha_O ( \Vert y_s^0(k) - y_s^* \Vert ).
	\end{split}
\end{equation} 
Because of the Lipschitz continuity of $g_x(\cdot)$, for which $\Vert x_s^0(k) - x_s^* \Vert \leq L_g \Vert y_s^0(k) - y_s^* \Vert$, we have that
 \begin{equation}
	\begin{split}
		W(x(k),y_t) \! \geq \! \alpha_\ell ( \Vert x(k) \! - \! x_s^0(k) \Vert ) \! + \! \alpha_O \left( \frac{1}{L_g} \Vert x_s^0(k) \! - \! x_s^* \Vert \right).
	\end{split}
\end{equation} 
Provided that $\alpha_\ell(\cdot)$ and $\alpha_O(\cdot)$ are $\mathcal{K}_\infty$ functions, the function $\alpha_{\ell O}(\cdot) \coloneqq \min \{ \alpha_\ell(\cdot), \allowbreak \alpha_O(\cdot) \}$ is also a $\mathcal{K}_\infty$ function and, thanks to standard properties of $\mathcal{K}_\infty$ functions, allows us to prove that
 \begin{equation}
	\begin{split}
		&W(x(k),y_t) \! \geq \! \alpha_{\ell O} ( \Vert x(k) \! - \! x_s^0(k) \Vert ) \! + \! \alpha_{\ell O} \left( \frac{1}{L_g} \Vert x_s^0(k) \! - \! x_s^* \Vert \right)\\
		&\geq \alpha_{\ell O} \Big( \frac{1}{2} \Vert x(k) - x_s^0(k) \Vert + \frac{1}{2L_g} \Vert x_s^0(k) - x_s^* \Vert \Big) \\
		&\geq \alpha_{\ell O} \left( \min \left\{ \frac{1}{2},\frac{1}{2L_g} \right\} \Vert x(k) - x_s^* \Vert \right).
	\end{split}
\end{equation} 
Defining $\underline{\alpha}(\cdot)$ as a $\mathcal{K}_\infty$ function such that $\underline{\alpha}(\Vert x \Vert) \leq \alpha_{\ell O} \left( \min \left\{ \frac{1}{2},\frac{1}{2L_g} \right\} \Vert x \Vert \right)$ for all $x \in \mathbb{R}^n$, we finally get that
 \begin{equation}
	\label{eq:rnmpct_W_lower_bound_xs*}
	W(x(k),y_t) \geq \underline{\alpha}(\Vert x(k) - x_s^* \Vert).
\end{equation} 

The upper bound of $W(x(k),y_t)$ can be obtained thanks to Assumption \ref{ass:rnmpct_robust_terminal_set}, which implies $V_f(x(k)-x_s^*,y_s^*) \leq \alpha_f(\Vert x(k) - x_s^* \Vert)$ in $\Omega(y_s^*)$. This, together with $\mathcal{Z}$ being a closed set, implies, due to standard arguments in the MPC literature, that there exists a $\mathcal{K}_\infty$ function $\bar{\alpha}(\cdot)$ such that
 \begin{equation}
	\label{eq:rnmpct_W_upper_bound_xs*}
	\begin{split}
		W(x(k),y_t) &\leq \bar{\alpha}(\Vert x(k) - x_s^* \Vert)
	\end{split}
\end{equation} 
for all $x(k) \in \mathcal{X}_{N_p}$.

We now need to prove that there exist a $\mathcal{K}_\infty$ function $\alpha(\cdot)$ and a $\mathcal{K}$ function $\lambda(\cdot)$ that verify
 \begin{equation}
	W(x(k+1),y_t) - W(x(k),y_t) \leq \lambda(\Vert w(k) \Vert) - \alpha (\Vert x(k) - x_s^* \Vert).
\end{equation} 
To this end, we first prove the desired descent property with respect to the equilibrium $x_s^0$ (which does not have to be equal to $x_s^*$), and then prove the descent property with respect to $x_s^*$.

Let us define $\hat{x}^0(j \vert k) = \phi_\pi(j;y_s^0(k),x(k),\mathbf{v}^0(k),\mathbf{0})$ and $\hat{x}^+(j \vert k+1) = \phi_\pi(j;y_s^+(k+1),x(k+1),\mathbf{v}^+(k+1),\mathbf{0})$. Then
 \begin{equation}
	\begin{split}
		&V_{N_p}(x(k+1),y_t,\mathbf{v}^+(k+1),y_s^+(k+1))\\ &=  \! \sum_{j=0}^{N_p-1}  \ell \big( \hat{x}^+(j \vert k \! + \! 1) \! - \! x_s^+(k \! + \! 1),v^+(j \vert k \! + \! 1) \! - \! v_s^+(k \! + \! 1) \big) \\
		&+ V_f\big(\hat{x}^+(N_p \vert k+1) - x_s^+(k+1), y_s^+(k+1)\big)\\ &+ V_O(y_s^+(k+1) - y_t).
	\end{split}
\end{equation} 
From the component-wise uniform continuity of the model, we have that, for all $i \in [1,n]$,
 \begin{equation} 
	\begin{split}
		\vert \hat{x}_i^+(0 \vert k+1) - \hat{x}_i^0(1 \vert k) \vert &\leq \sum_{a = 1}^r \sigma_{w,ia} (w_a(k))\eqqcolon c_{i,0}(w(k)).
	\end{split}
\end{equation} 
In the same way,
 \begin{equation} 
	\begin{split}
		\vert \hat{x}_i^+(1 \vert k+1) - \hat{x}_i^0(2 \vert k) \vert &\leq \sum_{b = 1}^n \sigma_{x,ib} \left( \sum_{a = 1}^r \sigma_{w,ba} (w_a(k)) \right)\\ &= \! \sum_{b = 1}^n \sigma_{x,ib} (c_{b,0}(w(k))) \\ &\eqqcolon  c_{i,1}(w(k)).
	\end{split}
\end{equation} 
Generalizing for all $j > 0$,
 \begin{equation} 
	\begin{split}
		\vert \hat{x}_i^+(j \vert k \! + \! 1) \! - \! \hat{x}_i^0(j \! + \! 1 \vert k) \vert \! &\leq \! \sum_{b = 1}^n \sigma_{x,ib} \left( c_{b,j-1}(w(k)) \right)\\ &\eqqcolon \! c_{i,j}(w(k)).
	\end{split}
\end{equation} 

Then, considering the uniform continuity of the stage cost function $\ell(\cdot,\cdot)$ and the uniform continuity of the terminal cost $V_f(\cdot,y_s)$ with respect to $x$, we obtain:
 \begin{equation}
	\begin{split}
		&W(x(k+1),y_t) - W(x(k),y_t)\\ &= V_{N_p}^0(x(k+1),y_t) - V_{N_p}^0(x(k),y_t) \\
		&\leq V_{N_p}(x(k+1),y_t, \mathbf{v}^+(k+1),y_s^+(k+1)) - V_{N_p}^0(x(k),y_t) \\
		&= \! \sum_{j=0}^{N_p-2} \Big[ \ell \big( \hat{x}^+(j \vert k \! + \! 1) \! - \! x_s^+(k \! + \!1), v^+(j \vert k \! + \! 1) \! - \! v_s^+(k+1) \big)\\ 
		&- \ell \big( \hat{x}^0(j+1 \vert k) - x_s^0(k), v^0(j+1 \vert k) - v_s^0(k) \big)  \Big] \\
		&- \ell \big( x(k) - x_s^0(k), v^0(0 \vert k) - v_s^0(k) \big)\\
		&+ V_f \big( \hat{x}^+(N_p \vert k+1) - x_s^+(k+1), y_s^+(k+1) \big)\\ &- V_f \big( \hat{x}^0(N_p \vert k) - x_s^0(k), y_s^0(k) \big)\\
		&+ \! \ell \big( \hat{x}^+(N_p \! - \! 1 \vert k \!+ \! 1) \! - \! x_s^+(k \! + \! 1), v_f(y_s^+(k \! + \!1)) \! - \! v_s^+(k \! + \! 1) \big).
	\end{split}
\end{equation} 
Analyzing the stage cost function, we find that
 \begin{equation}
	\begin{split}
		&\ell \big( \hat{x}^+(j \vert k+1) - x_s^+(k+1), v^+(j \vert k+1) - v_s^+(k+1) \big)\\ &- \ell \big( \hat{x}^0(j  +  1 \vert k)  -  x_s^0(k), v^0(j+1 \vert k)  -  v_s^0(k) \big)\\ 
		&\leq \sigma_{\ell,x} \big( \Vert [ c_{1,j}(w(k)) \; \hdots \; c_{n,j}(w(k)) ]^\top \Vert \big).
	\end{split}
\end{equation} 
Defining $r_{j}(w) \coloneqq \Vert [c_{1,j}(w) \; \hdots \; c_{n,j}(w)]^\top \Vert$, the last inequality becomes
 \begin{equation}
	\begin{split}
		&\ell \big( \hat{x}^+(j \vert k+1) - x_s^+(k+1), v^+(j \vert k+1) - v_s^+(k+1) \big)\\ &- \ell \big( \hat{x}^0(j + 1 \vert k) - x_s^0(k), v^0(j+1 \vert k) - v_s^0(k) \big)\\ 
		&\leq \sigma_{\ell,x} \circ r_{j}(w(k)).
	\end{split}
\end{equation} 
Moreover, analyzing the terminal cost we get that
 \begin{equation}
	\begin{split}
		&V_f \big( \hat{x}^+(N_p-1 \vert k+1)  -  x_s^+(k+1), y_s^+(k+1) \big)\\  &- V_f \big( \hat{x}^0(N_p \vert k)  -  x_s^0(k), y_s^0(k) \big) \leq \sigma_f \circ r_{N_p-1}(w(k)).
	\end{split}
\end{equation} 
Due to the choice $y_s^+(k+1) = y_s^0(k)$, which implies $x_s^+(k+1) = x_s^0(k)$, and given that 
 \begin{equation}
	\begin{split}
		&V_f \big( \hat{x}^+(N_p \vert k+1) - x_s^+(k+1), y_s^+(k+1) \big)\\ &- V_f \big( \hat{x}^0(N_p \vert k) - x_s^0(k), y_s^0(k) \big)\\ &= V_f \big( \hat{x}^+(N_p \vert k+1) - x_s^+(k+1), y_s^+(k+1) \big)\\ &- V_f \big( \hat{x}^+(N_p-1 \vert k+1) - x_s^+(k+1), y_s^+(k+1) \big)\\ &+ V_f \big( \hat{x}^+(N_p-1 \vert k+1) - x_s^+(k+1), y_s^+(k+1) \big)\\ &- V_f \big( \hat{x}^0(N_p \vert k) - x_s^0(k), y_s^0(k) \big),
	\end{split}
\end{equation} 
we get
 \begin{equation}
	\begin{split}
		\Delta W &\coloneqq W(x(k   +   1),y_t)   -   W(x(k),y_t)\\ &\leq \sum_{j=0}^{N_p-2} \sigma_{\ell,x} \circ r_{j}(w(k)) + \sigma_f \circ r_{N_p-1}(w(k))\\ &- \ell \big( x(k) - x_s^0(k), v^0(0 \vert k) - v_s^0(k) \big)\\
		&+ \ell \big( \hat{x}^+(N_p-1 \vert k+1) - x_s^+(k+1),\\ &v_f(y_s^+(k+1))
		- v_s^+(k+1) \big)\\ &+ V_f \big( \hat{x}^+(N_p \vert k+1) - x_s^+(k+1), y_s^+(k+1) \big)\\ &- V_f \big( \hat{x}^+(N_p-1 \vert k+1) - x_s^0+(k+1), y_s^+(k+1) \big).
	\end{split}
\end{equation} 
Thanks to Assumption \ref{ass:rnmpct_robust_terminal_set}, we have that
 
\begin{align}
	&V_f(\hat{x}^+(N_p \vert k  +  1)   -   x_s^+(k+1),y_s^+(k+1)  ) \\   &-   V_f(\hat{x}^+(N_p   -  1 \vert k  + 1)  -   x_s^+(k+1),y_s^+(k+1) ) \nonumber \\ &\leq - \ell\big( \hat{x}^+(N_p \! - \! 1 \vert k \! + \! 1) \! - \! x_s^+(k \! + \!1),v_f(y_s^+(k \! + \! 1)) \! - \! v_s^+(k \! + \! 1) \big) \nonumber
\end{align}

in $\Omega(y_s^+(k+1))$, and hence, given that the functions $r_j(w)$ are vector norms, there exists a $\mathcal{K}$ function $\lambda_1(\cdot)$ such that 
 \begin{equation}
	\Delta W \leq - \ell \big( x(k)-x_s^0(k),v^0(0 \vert k) - v_s^0(k) \big) + \lambda_1(\Vert w(k) \Vert).
\end{equation} 
Given that $\ell(x,v) \geq \alpha_\ell (\Vert x \Vert)$, the last finding results in
 \begin{equation}
	\label{eq:rnmpct_W_decrease_xs0}
	\Delta W \leq - \alpha_\ell (\Vert x(k) - x_s^0(k) \Vert) + \lambda_1(\Vert w(k) \Vert),
\end{equation} 
which proves the descent property of $W(\cdot,y_t)$ with respect to the equilibrium $x_s^0$. This result will be exploited to prove the nominal descent property of $W(\cdot,y_t)$ with respect to $x_s^*$, which is needed to prove that $W(\cdot,y_t)$ is an ISS-Lyapunov function. In particular, we analyze two different cases. The logic behind the two cases is the following, and is inspired by \cite{Kohler2020_periodic}: 
\begin{itemize}
	\item if the tracking error is big, namely $x(k)$ is far from $x_s^0(k)$, the nominal descent property of $W(\cdot,y_t)$ with respect to $x_s^0(k)$ implies that $W(\cdot,y_t)$ nominally decreases also with respect to $x_s^*$;
	\item if at time instant $k$ the tracking error is small and $y_s^0(k) \neq y_s^*$, namely $x(k)$ is close to $x_s^0(k)$, but the equilibrium $x_s^0(k)$ is different from the best admissible steady state $x_s^*$, at time instant $k+1$ it is always possible to find a feasible candidate solution $(\tilde{\mathbf{v}}(k+1),\tilde{y}_s(k+1))$ such that $V_O(\tilde{y}_s(k+1) - y_t) < V_O(y_s^0(k) - y_t)$, and the decrease in $V_O(\cdot)$ also ensures the nominal descent of $W(\cdot,y_t)$ with respect to $x_s^*$.
\end{itemize}
To formally distinguish between the two cases, we analyze the ratio between the contributions of the tracking error and the offset cost function to the total cost. In detail, given a real constant $\gamma > 0$, which will be specified later, having a large tracking error translates into $\alpha_\ell(\Vert x(k) - x_s^0(k) \Vert) \geq \gamma \alpha_O(\Vert y_s^0(k) - y_s^* \Vert)$ (case 1), while having a tracking error contribution that is small compared to the offset cost function contribution results in $\alpha_\ell(\Vert x(k) - x_s^0(k) \Vert) \leq \gamma \alpha_O(\Vert y_s^0(k) - y_s^* \Vert)$ (case 2). We now prove the nominal descent property of $W(\cdot,y_t)$ in the two cases.

\textbf{Case 1:} Provided that 
 \begin{equation}
	\label{eq:rnmpct_case_1}
	\alpha_\ell(\Vert x(k) - x_s^0(k) \Vert) \geq \gamma \alpha_O(\Vert y_s^0(k) - y_s^* \Vert),
\end{equation} 
\eqref{eq:rnmpct_W_decrease_xs0} and the Lipschitz continuity of $g_x(\cdot)$ imply
 \begin{equation}
	\label{eq:rnmpct_W_decrease_case_1_a}
	\begin{split}
		&\Delta W \leq \lambda_1(\Vert w(k) \Vert) - \alpha_\ell(\Vert x(k) - x_s^0(k) \Vert)\\
		&= \lambda_1(\Vert w(k) \Vert) \! - \! \frac{1}{2}\alpha_\ell(\Vert x(k) \! - \! x_s^0(k) \Vert) \! - \! \frac{1}{2}\alpha_\ell(\Vert x(k) \! - \! x_s^0(k) \Vert) \\
		&\leq \lambda_1(\Vert w(k) \Vert) \! - \! \frac{1}{2} \Big[ \alpha_\ell(\Vert x(k) \! - \! x_s^0(k) \Vert) \! + \! \gamma \alpha_O(\Vert y_s^0(k) \! - \! y_s^* \Vert) \Big]\\
		&\leq \lambda_1(\Vert w(k) \Vert) - \frac{1}{2} \Bigg[ \alpha_\ell(\Vert x(k) - x_s^0(k) \Vert)\\
		&+ \gamma \alpha_O\left(\frac{1}{L_g}\Vert x_s^0(k) - x_s^* \Vert\right) \Bigg]\\
		&\leq \lambda_1(\Vert w(k) \Vert) - \frac{1}{2} \min \{1, \gamma\} \Bigg[ \alpha_\ell(\Vert x(k) - x_s^0(k) \Vert)\\
		&+ \alpha_O\left(\frac{1}{L_g}\Vert x_s^0(k) - x_s^* \Vert\right) \Bigg].
	\end{split}
\end{equation} 
Using the functions $\alpha_{\ell O}(\cdot)$ and $\underline{\alpha}(\cdot)$ defined earlier, we get
\begin{equation}
	\label{eq:rnmpct_W_decrease_case_1_b}
	\begin{split}
		&\Delta W \leq \lambda_1(\Vert w(k) \Vert) - \frac{1}{2} \min \{1, \gamma\} \Bigg[ \alpha_{\ell O}(\Vert x(k) - x_s^0(k) \Vert)\\
		&+ \alpha_{\ell O}\left(\frac{1}{L_g}\Vert x_s^0(k) - x_s^* \Vert\right) \Bigg] \\
		&\leq \lambda_1(\Vert w(k) \Vert) - \frac{1}{2} \min \{1, \gamma\} \alpha_{\ell O} \Bigg( \frac{1}{2} \Vert x(k) - x_s^0(k) \Vert\\
		&+ \frac{1}{2L_g} \Vert x_s^0(k) - x_s^* \Vert \Bigg) \\
		&\leq \lambda_1(\Vert w(k) \Vert) \! - \! \frac{1}{2} \min \{1, \gamma\}
		\alpha_{\ell O} \!\! \left(\! \min \left\{ \! \frac{1}{2},\frac{1}{2L_g} \! \right\} \! \Vert x(k) \! - \! x_s^* \Vert \! \right)\\
		&\leq \lambda_1(\Vert w(k) \Vert) - \frac{1}{2} \min \{1, \gamma\} \underline{\alpha}(\Vert x(k) - x_s^* \Vert).
	\end{split}
\end{equation}
\textbf{Case 2:} Assume
\begin{equation}
	\label{eq:rnmpct_case_2}
	\alpha_\ell(\Vert x(k) - x_s^0(k) \Vert) \leq \gamma \alpha_O(\Vert y_s^0(k) - y_s^* \Vert).
\end{equation} 
The proof of the nominal descent of $W(\cdot,y_t)$ in this case is divided into five steps:
\begin{enumerate}
	\item We first prove that, by choosing a small enough $\gamma$, \eqref{eq:rnmpct_case_2} implies that $x(k) \in \mathcal{X}_f(y_s^0(k))$; in particular, $\alpha_\ell(\Vert x(k)-x_s^0(k) \Vert) \leq \epsilon_x$, where $\epsilon_x > 0$ is a real constant, whose meaning will be clarified later.
	\item We prove that, given a state $x(k)$ such that $\alpha_\ell(\Vert x(k)-x_s^0(k) \Vert) \leq \epsilon_x$, the following nominal state $\hat{x}^0(1 \vert k)$ is such that $\alpha_\ell(\Vert \hat{x}^0(1 \vert k) - x_s^0(k) \Vert) \leq \epsilon_x / 2$.
	\item Leveraging the fact that $\hat{x}^0(1 \vert k)$ is in the interior of $\mathcal{X}_f(y_s^0(k))$, we prove that a candidate artificial reference $\tilde{y}_s(k+1)$ such that $ V_O(\tilde{y}_s(k+1) - y_t) < V_O(y_s^0(k) - y_t)$ would be feasible for $\hat{x}^0(1 \vert k)$.
	\item We prove that the candidate artificial reference $\tilde{y}_s(k+1)$ makes the function $W(\cdot,y_t)$ decrease in the perturbation-free case.
	\item We prove that the gap $W(x(k+1)) - W(\hat{x}^0(1 \vert k))$, where $x(k+1)$ is the uncertain following state, is bounded from above by a $\mathcal{K}$ function $\lambda_2: \; \mathcal{W} \to \mathbb{R}_{\geq 0}$, which allows us to state that $W(\cdot,y_t)$ is an ISS-Lyapunov function.
\end{enumerate}
Let us first discuss the role of the constant $\epsilon_x$. Assuming to have terminal regions $\mathcal{X}_f(y_s)$ such that $\{x_s\} \subset \mathcal{X}_f(y_s)$, namely the sets $\mathcal{X}_f(y_s)$ contain more than one point, there surely exists a constant $\epsilon_x > 0$ such that the condition $\alpha_\ell(\Vert x - x_s \Vert) \leq \epsilon_x$ implies $x \in \mathcal{X}_f(y_s)$, for all $y_s \in \mathcal{Y}_t$.
Moreover, provided that $\alpha_\ell(\Vert x - x_s \Vert) \leq \epsilon_x$ implies $x \in \mathcal{X}_f(y_s)$, we also have that, for all $x$ that verify this condition, the couple $\big( \mathbf{v},y_s \big)$, with $\mathbf{v} \coloneqq \{ v_f(y_s), \hdots, v_f(y_s) \}$, is a feasible solution and the next inequality holds true:
\begin{equation}
	\label{eq:rnmpct_Vf_bound}
	\begin{split}
		V^0_{N_p}(x,y_t) &\leq V_f(x - x_s, y_s) + V_O(y_s - y_t) \\
		&\leq \alpha_f(\Vert x - x_s \Vert) + V_O(y_s - y_t).
	\end{split}
\end{equation} 
Provided that the set $\mathcal{Y}_t$ is bounded, there exists a real constant $\alpha_{O,\mathrm{max}} > 0$ such that $\alpha_O(\Vert y_1 - y_2 \Vert) \leq \alpha_{O,\mathrm{max}}$ for any $y_1, y_2 \in \mathcal{Y}_t$. Therefore, by choosing $\gamma \leq \gamma_1 \coloneqq \frac{\epsilon_x}{\alpha_{O,\mathrm{max}}}$, we get that
\begin{equation}
	\alpha_\ell(\Vert x(k) - x_s^0(k) \Vert) \leq \gamma \alpha_{O,\mathrm{max}} \leq \epsilon_x,
\end{equation} 
i.e. $x(k) \in \mathcal{X}_f(y_s^0(k))$.

Now, given the optimal solution $(\mathbf{v}^0(k),y_s^0(k))$, we use the previously defined notation $\hat{x}^0(j \vert k) = \phi_\pi(j;y_s^0(k),x(k),\mathbf{v}^0(k),\mathbf{0})$ and focus on the following nominal state $\hat{x}^0(1 \vert k)$. Specifically, we note that \eqref{eq:rnmpct_case_2} and \eqref{eq:rnmpct_Vf_bound} imply 
\begin{equation}
	\label{eq:rnmpct_1_step_inside_a}
	\begin{split}
		\alpha_\ell(\Vert \hat{x}^0(1 \vert k) \! - \! x_s^0(k) \Vert) &\leq \ell \big( \hat{x}^0(1 \vert k) \! - \! x_s^0(k), v^0(1 \vert k) \! - \! v_s^0(k) \big)\\
		&\leq \alpha_f (\Vert x(k) \! - \! x_s^0(k) \Vert).
	\end{split}
\end{equation} 
Since $\alpha_\ell (\Vert x(k) - x_s^0(k) \Vert) \leq \gamma \alpha_{O,\mathrm{max}}$ and $\alpha_\ell(\cdot)$ is a $\mathcal{K}_\infty$ function, it also holds true that $\Vert x(k) - x_s^0(k) \Vert \leq \alpha_\ell^{-1}(\gamma \alpha_{O,\mathrm{max}})$, which implies that
\begin{equation}
	\label{eq:rnmpct_1_step_inside_b}
	\begin{split}
		\alpha_\ell(\Vert \hat{x}^0(1 \vert k) - x_s^0(k) \Vert) &\leq \alpha_f (\Vert x(k) - x_s^0(k) \Vert) \\
		&\leq \alpha_f \circ \alpha_\ell^{-1} (\gamma \alpha_{O,\mathrm{max}}).
	\end{split}
\end{equation} 
For $\gamma \leq \gamma_2 \coloneqq \frac{\alpha_\ell \circ \alpha_f^{-1} ( \epsilon_x / 2 )}{\alpha_{O,\mathrm{max}}}$, we have that $\alpha_\ell(\Vert \hat{x}^0(1 \vert k) - x_s^0(k) \Vert) \leq \epsilon_x / 2$. This means that the nominal state $\hat{x}^0(1 \vert k)$ is inside the interior of $\left\{ x: \; \alpha_\ell(\Vert x - x_s^0 \Vert) \leq \epsilon_x \right\}$, which would make the problem $P_{N_p}(\hat{x}^0(1 \vert k),y_t)$ feasible, as $\hat{x}^0(1 \vert k) \in \mathcal{X}_f(y_s^0(k))$ and the terminal control input $v_f(y_s^0(k))$ is feasible. 

Given that $\hat{x}^0(1 \vert k)$ is in the interior of $\mathcal{X}_f(y_s^0(k))$, we can consider the possibility of selecting a new artificial reference $\tilde{y}_s(k+1)$, closer to $y_s^*$, while still ensuring that the state remains within the set $\mathcal{X}_f(\tilde{y}_s(k+1))$. In detail, we seek a candidate $\tilde{y}_s(k+1)$ such that $V_O(\tilde{y}_s(k+1) - y_t) < V_O(y_s^+(k+1) - y_t)$, recalling that $y_s^+(k+1) \equiv y_s^0(k)$. 
Exploiting the convexity of $V_O(\cdot)$, there exists a new candidate $\tilde{y}_s(k+1) \in \mathcal{Y}_t$ such that
\begin{equation}
	\label{eq:rnmpct_candidate_hat_ys}
	\tilde{y}_s(k+1) \coloneqq \beta y_s^0(k) + (1-\beta)y_s^*, \quad \beta \in [0,1].
\end{equation} 
Equation \eqref{eq:rnmpct_candidate_hat_ys} can be rewritten as 
\begin{equation}
	\label{eq:rnmpct_candidate_hat_ys_2}
	\tilde{y}_s(k+1) - y_s^0(k) = (1 - \beta) (y_s^* - y_s^0(k)).
\end{equation} 
Due to this new candidate solution, the offset cost satisfies
\begin{equation}
	\label{eq:rnmpct_V_O_decrease_1}
	\begin{split}
		&V_O(\tilde{y}_s(k+1) - y_t) - V_O(y_s^0(k) - y_t)\\
		&= V_O\big(\beta (y_s^0(k) - y_t) + (1 - \beta)(y_s^* - y_t) \big) - V_O(y_s^0(k) - y_t).
	\end{split}
\end{equation} 
Given the convexity of $V_O(\cdot)$, \eqref{eq:rnmpct_V_O_decrease_1} becomes
 \begin{equation}
	\label{eq:rnmpct_V_O_decrease_2}
	\begin{split}
		&V_O(\tilde{y}_s(k+1) - y_t) - V_O(y_s^0(k) - y_t)\\
		&\leq \beta V_O\big(y_s^0(k) \! - \! y_t\big) \! + \! (1 \! - \! \beta) V_O\big(y_s^* \! - \! y_t\big) - V_O(y_s^0(k) \! - \! y_t)\\
		&= -(1 - \beta) \Big[ V_O\big(y_s^0(k) - y_t\big) - V_O\big(y_s^* - y_t\big) \Big] \\
		&\leq -(1 - \beta) \alpha_O (\Vert y_s^0(k) - y_s^* \Vert),
	\end{split}
\end{equation} 
which, for $\beta \in [0,1)$, proves that $\tilde{y}_s(k+1)$ is such that $V_O(\tilde{y}_s(k+1) - y_t) < V_O(y_s^0(k) \allowbreak - y_t)$. We now need to find a lower bound $\beta_1$ such that, for all $\beta \geq \beta_1$, the new candidate $\tilde{y}_s(k+1)$ ensures also that $\alpha_\ell(\Vert \hat{x}^0(1 \vert k) - \tilde{x}_s(k+1) \Vert) \leq \epsilon_x$. In order to find $\beta_1$, we start by stating that, because of the Lipschitz continuity of $g_x(\cdot)$, 
 \begin{equation}
	\label{eq:rnmpct_new_candidate_in_epsilon_a}
	\begin{split}
		&\alpha_\ell \big(\Vert \hat{x}^0(1 \vert k) - \tilde{x}_s(k+1) \Vert\big) \leq \alpha_\ell\big(\Vert \hat{x}^0(1 \vert k) - x_s^0(k) \Vert\\ &+ \Vert x_s^0(k) - \tilde{x}_s(k+1) \Vert\big) \\
		&\leq \alpha_\ell\big( \Vert \hat{x}^0(1 \vert k) - x_s^0(k) \Vert + L_g \Vert y_s^0(k) - \tilde{y}_s(k+1) \Vert \big) \\
		&= \alpha_\ell\big( \Vert \hat{x}^0(1 \vert k) - x_s^0(k) \Vert + L_g(1   -   \beta) \Vert y_s^0(k)   -   y_s^* \Vert \big) \\
		&\leq \alpha_\ell \big( \alpha_\ell^{-1} (\epsilon_x / 2) + L_g(1 - \beta) \Vert y_s^0(k) - y_s^* \Vert \big).
	\end{split}
\end{equation} 
Recall that $\alpha_O(\Vert y_s^0(k) - y_s^* \Vert) \leq \alpha_{O,\mathrm{max}}$, which implies $\Vert y_s^0(k) - y_s^* \Vert \leq \alpha_O^{-1} (\alpha_{O,\mathrm{max}})$. For this reason, \eqref{eq:rnmpct_new_candidate_in_epsilon_a} can be rewritten as
 \begin{equation}
	\label{eq:rnmpct_new_candidate_in_epsilon_b}
	\begin{split}
		&\alpha_\ell \big(\Vert \hat{x}^0(1 \vert k) - \tilde{x}_s(k+1) \Vert\big)\\
		&\leq \alpha_\ell \big( \alpha_\ell^{-1} (\epsilon_x / 2) + L_g(1 - \beta) \alpha_O^{-1}(\alpha_{O,\mathrm{max}}) \big).
	\end{split}
\end{equation} 
Therefore, for $\beta \geq \beta_1 \coloneqq 1 - \frac{\alpha_\ell^{-1}(\epsilon_x) - \alpha_\ell^{-1} (\epsilon_x / 2)}{L_g \cdot \alpha_O^{-1}(\alpha_{O,\mathrm{max}})}$, $\alpha_\ell(\Vert \hat{x}^0(1 \vert k) - \tilde{x}_s(k+1) \Vert) \leq \epsilon_x$, which ensures that:
\begin{enumerate}
	\item $\hat{x}^0(1 \vert k) \in \mathcal{X}_f(\tilde{y}_s(k+1))$, guaranteeing that there exists a candidate control sequence $\tilde{\mathbf{v}}(k+1)$ such that the couple $(\tilde{\mathbf{v}}(k+1),\tilde{y}_s(k+1))$ is a feasible solution to $P_{N_p}(\hat{x}^0(1 \vert k),y_t)$; indeed, $\tilde{\mathbf{v}}(k+1) \coloneqq \{ v_f(\tilde{y}_s(k+1)),\hdots, \allowbreak v_f(\tilde{y}_s(k+1)) \}$ is a feasible control sequence.
	\item The optimal control sequence $\tilde{\mathbf{v}}(k+1)$ is such that 
	 \begin{equation}
		\begin{split}
			&V_{N_p}^0(\hat{x}^0(1 \vert k),y_t,\tilde{\mathbf{v}}(k \! + \! 1),\tilde{y}_s(k \! + \!1)) \! - \! V_O(\tilde{y}_s(k \!+ \!1) \! - \! y_t)\\ &\leq \alpha_f (\Vert \hat{x}^0(1 \vert k) - \tilde{x}_s(k+1) \Vert).
		\end{split}
	\end{equation} 
\end{enumerate}
Due to $\alpha_\ell(\Vert \hat{x}^0(1 \vert k) - \tilde{x}_s(k+1) \Vert) \leq \epsilon_x$ and to the Lipschitz continuity of $g_x(\cdot)$, we have that
 \begin{equation}
	\label{eq:rnmpct_V_hat_bound}
	\begin{split}
		&V_{N_p}(\hat{x}^0(1 \vert k),y_t,\tilde{\mathbf{v}}(k+1),\tilde{y}_s(k+1))  -  V_O(\tilde{y}_s  -  y_t)\\ &\leq \alpha_f(\Vert \hat{x}^0(1 \vert k)  -  \tilde{x}_s(k+1) \Vert) \\
		&\leq \alpha_f \big( \Vert \hat{x}^0(1 \vert k) - x_s^0(k) \Vert + \Vert x_s^0(k) - \tilde{x}_s(k+1) \Vert \big) \\
		&\leq \alpha_f \big( \Vert \hat{x}^0(1 \vert k) - x_s^0(k) \Vert + L_g\Vert y_s^0(k) - \tilde{y}_s(k+1) \Vert \big) \\
		&= \alpha_f \big( \Vert \hat{x}^0(1 \vert k) - x_s^0(k) \Vert + L_g(1 - \beta)\Vert y_s^0(k) - y_s^* \Vert \big).
	\end{split}
\end{equation} 
To analyze $W(x(k+1),y_t) - W(x(k),y_t) = V^0_{N_p}(x(k+1),y_t) \allowbreak - V^0_{N_p}(x(k),y_t)$, where $x(k+1)$ is the following uncertain state, we note that
 \begin{equation}
	\label{eq:rnmpct_DeltaW_split}
	\begin{split}
		&V^0_{N_p}(x(k+1),y_t) - V^0_{N_p}(x(k),y_t)\\ 
		&\leq V_{N_p}(x(k+1),y_t,\tilde{\mathbf{v}}(k+1),\tilde{y}_s(k+1)) - V^0_{N_p}(x(k),y_t) \\ 
		&\leq V_{N_p}(x(k+1),y_t,\tilde{\mathbf{v}}(k+1),\tilde{y}_s(k+1))\\ &- V_{N_p}(\hat{x}^0(1 \vert k),y_t,\tilde{\mathbf{v}}(k+1),\tilde{y}_s(k+1)) \\ 
		&+ V_{N_p}(\hat{x}^0(1 \vert k),y_t,\tilde{\mathbf{v}}(k+1),\tilde{y}_s(k+1)) - V^0_{N_p}(x(k),y_t).
	\end{split}
\end{equation} 
We analyze this difference in two steps: (i) first, we analyze the nominal descent of $W(\cdot,y_t)$, which is provided by the term $V_{N_p}(\hat{x}^0(1 \vert k),y_t,\tilde{\mathbf{v}}(k+1),\tilde{y}_s(k+1)) - V^0_{N_p}(x(k),y_t)$, (ii) then, we analyze the contribution of the perturbation $w(k)$, which is provided by $V_{N_p}(x(k+1),y_t;\tilde{\mathbf{v}}(k+1),\tilde{y}_s(k+1)) - V_{N_p}(\hat{x}^0(1 \vert k),y_t;\tilde{\mathbf{v}}(k+1), \allowbreak \tilde{y}_s(k+1))$.

Let us analyze $V_{N_p}(\hat{x}^0(1 \vert k),y_t,\tilde{\mathbf{v}}(k+1),\tilde{y}_s(k+1)) - V^0_{N_p}(x(k),y_t)$:
 \begin{equation}
	\label{eq:rnmpct_W_decrease_xs*}
	\begin{split}
		&V_{N_p}(\hat{x}^0(1 \vert k),y_t,\tilde{\mathbf{v}}(k+1),\tilde{y}_s(k+1)) - V^0_{N_p}(x(k),y_t)\\ 
		&\leq \alpha_f(\Vert \hat{x}^0(1 \vert k) - \tilde{x}_s(k+1) \Vert) - \alpha_\ell(\Vert x(k) - x_s^0(k) \Vert) \\
		&+ V_O(\tilde{y}_s(k+1) - y_t) - V_O(y_s^0(k) - y_t).
	\end{split}
\end{equation} 
Equations \eqref{eq:rnmpct_V_O_decrease_2} and \eqref{eq:rnmpct_V_hat_bound} imply that \eqref{eq:rnmpct_W_decrease_xs*} becomes
 \begin{equation}
	\label{eq:rnmpct_W_decrease_xs*_2}
	\begin{split}
		&V_{N_p}(\hat{x}^0(1 \vert k),y_t,\tilde{\mathbf{v}}(k+1),\tilde{y}_s(k+1)) - V^0_{N_p}(x(k),y_t)\\
		&\leq \alpha_f \big( \Vert \hat{x}^0(1 \vert k)  -  x_s^0(k) \Vert  +  L_g(1  -  \beta) \Vert y_s^0(k)  -  y_s^* \Vert \big)\\ &-  \alpha_\ell(\Vert x(k)  -  x_s^0(k) \Vert)-(1 - \beta) \alpha_O(\Vert y_s^0(k) - y_s^* \Vert) \\
		&\leq \alpha_f\big(2 \Vert \hat{x}^0(1 \vert k)  -  x_s^0(k) \Vert\big)\\ &+ \alpha_f\big( 2L_g(1 - \beta) \Vert y_s^0(k) - y_s^* \Vert \big) \\
		&-(1 - \beta) \alpha_O(\Vert y_s^0(k) - y_s^* \Vert) - \alpha_\ell(\Vert x(k) - x_s^0(k) \Vert).
	\end{split}
\end{equation} 
Recalling that $\alpha_f(\Vert x \Vert) = \bar{\alpha}_f \cdot \Vert x \Vert^2$, \eqref{eq:rnmpct_W_decrease_xs*_2} rewrites as
 \begin{equation}
	\label{eq:rnmpct_W_decrease_xs*_2b}
	\begin{split}
		&V_{N_p}(\hat{x}^0(1 \vert k),y_t,\tilde{\mathbf{v}}(k+1),\tilde{y}_s(k+1)) - V^0_{N_p}(x(k),y_t)\\
		&\leq 4 \bar{\alpha}_f \Vert \hat{x}^0(1 \vert k)  -  x_s^0(k) \Vert^2 + 4\bar{\alpha}_f \cdot L_g^2 (1 - \beta)^2 \Vert y_s^0(k) - y_s^* \Vert^2 \\
		&-(1 - \beta) \alpha_O(\Vert y_s^0(k) - y_s^* \Vert) - \alpha_\ell(\Vert x(k) - x_s^0(k) \Vert).
	\end{split}
\end{equation} 
Let us define the function $g: \; \mathbb{R} \times \mathcal{S} \to \mathbb{R}$ as 
 \begin{equation}
	\begin{split}
		g(\beta,s) &\coloneqq (1-\beta) \alpha_O(s)- 4\bar{\alpha}_f L_g^2 (1-\beta)^2 s^2,
	\end{split}
\end{equation} 
such that \eqref{eq:rnmpct_W_decrease_xs*_2b} becomes
 \begin{equation}
	\label{eq:rnmpct_W_decrease_xs*_2c}
	\begin{split}
		&V_{N_p}(\hat{x}^0(1 \vert k),y_t,\tilde{\mathbf{v}}(k+1),\tilde{y}_s(k+1)) - V^0_{N_p}(x(k),y_t) \\
		&\leq   4 \bar{\alpha}_f \Vert \hat{x}^0(1 \vert k)  -  x_s^0(k) \Vert^2    -g(\beta, \Vert y_s^0(k)   -  y_s^* \Vert)\\ &- \alpha_\ell(\Vert x(k)   -   x_s^0(k) \Vert).
	\end{split}
\end{equation} 
In order to prove the nominal descent property of $W(\cdot,y_t)$, we must ensure that there exists a $\mathcal{K}_\infty$ function $\alpha_\beta(\cdot)$ such that $g(\beta,\Vert y_s^0(k) - y_s^* \Vert) \geq \alpha_\beta(\Vert y_s^0(k) - y_s^* \Vert)$. Thanks to Assumption \ref{ass:rnmpct_alpha_O_alpha_f_ratio}, in order for these conditions to hold true, it is sufficient to choose $\beta \in (\beta_2,1)$, where $\beta_2 \coloneqq 1 - \min \{\mathfrak{b}_1,\mathfrak{b}_2\}$. Since both $\mathfrak{b}_1$ and $\mathfrak{b}_2$ are assumed to be positive, the interval $(\beta_2,1)$ is guaranteed to be non-empty. This is shown at the end of the current proof.

Given the function $\alpha_\beta(\cdot)$, \eqref{eq:rnmpct_W_decrease_xs*_2c} can be rewritten as
 \begin{equation}
	\label{eq:rnmpct_W_decrease_xs*_3}
	\begin{split}
		&V_{N_p}(\hat{x}^0(1 \vert k),y_t,\tilde{\mathbf{v}}(k+1),\tilde{y}_s(k+1)) - V^0_{N_p}(x(k),y_t) \\
		&\leq   4 \bar{\alpha}_f \Vert \hat{x}^0(1 \vert k)  -  x_s^0(k) \Vert^2   -\alpha_\beta(\Vert y_s^0(k)   -   y_s^* \Vert)\\ &-   \alpha_\ell(\Vert x(k)   -   x_s^0(k) \Vert)\\
		&\leq 4 \bar{\alpha}_f \Vert \hat{x}^0(1 \vert k)  -  x_s^0(k) \Vert^2 -\frac{1}{2}\alpha_\beta(\Vert y_s^0(k) - y_s^* \Vert)\\
		&-\frac{1}{2}\alpha_\beta(\Vert y_s^0(k) - y_s^* \Vert) - \alpha_\ell(\Vert x(k) - x_s^0(k) \Vert).
	\end{split}
\end{equation} 
We now analyze \eqref{eq:rnmpct_W_decrease_xs*_3} in two steps:
\begin{enumerate}
	\item we first prove that it is always possible to choose a value for $\gamma$ that is sufficiently small to ensure that $4 \bar{\alpha}_f \Vert \hat{x}^0(1 \vert k)  -  x_s^0(k) \Vert^2 -\frac{1}{2}\alpha_\beta(\Vert y_s^0(k) - y_s^* \Vert) \leq 0$;
	\item we then prove that the term $\frac{1}{2}\alpha_\beta(\Vert y_s^0(k) - y_s^* \Vert) + \alpha_\ell(\Vert x(k) - x_s^0(k) \Vert)$ can be lower-bounded by a $\mathcal{K}_\infty$ function $\alpha_2(\Vert x(k) - x_s^* \Vert)$, which is then used to prove that $W(\cdot,y_t)$ is an ISS-Lyapunov function.
\end{enumerate}
Starting from the first point, recall that $\alpha_\ell(\Vert \hat{x}^0(1 \vert k) - x_s^0(k) \Vert) \leq \bar{\alpha}_f \Vert x(k) - x_s^0(k) \Vert^2$ and also that $\alpha_\ell (\Vert x(k) - x_s^0(k) \Vert) \leq \gamma \alpha_O (\Vert y_s^0(k) - y_s^* \Vert)$. Due to these inequalities, we get
 \begin{equation}
	\Vert \hat{x}^0(1 \vert k) - x_s^0(k) \Vert \leq \alpha_\ell^{-1} \Big( \bar{\alpha}_f \cdot \Big[ \alpha_\ell^{-1} \big( \gamma \alpha_O (\Vert y_s^0(k) - y_s^* \Vert) \big) \Big]^2 \Big),
\end{equation} 
and, as a consequence, $4 \bar{\alpha}_f \Vert \hat{x}^0(1 \vert k)  -  x_s^0(k) \Vert^2 -\frac{1}{2}\alpha_\beta(\Vert y_s^0(k) - y_s^* \Vert) \leq 0$ for $\gamma \leq \gamma_3$, where $\gamma_3$ is defined as
 \begin{equation}
	\gamma_3 \coloneqq \frac{\alpha_\ell \left( \sqrt{\frac{1}{\bar{\alpha}_f}\cdot \alpha_\ell \left( \sqrt{\frac{1}{8\bar{\alpha}_f} \alpha_\beta(\Vert y_s^0(k) - y_s^* \Vert)} \right)} \right)}{\alpha_O(\Vert y_s^0(k) - y_s^* \Vert)}.
\end{equation} 
Consequently, for $\gamma \leq \gamma_3$, and considering the Lipschitz continuity of $g_x(\cdot)$, \eqref{eq:rnmpct_W_decrease_xs*_3} can be rewritten as
 \begin{equation}
	\label{eq:rnmpct_W_decrease_xs*_4}
	\begin{split}
		&V_{N_p}(\hat{x}^0(1 \vert k),y_t,\tilde{\mathbf{v}}(k+1),\tilde{y}_s(k+1)) - V^0_{N_p}(x(k),y_t)\\ 
		&\leq -\frac{1}{2}\alpha_\beta(\Vert y_s^0(k) - y_s^* \Vert) - \alpha_\ell(\Vert x(k) - x_s^0(k) \Vert)\\
		&\leq -\frac{1}{2}\alpha_\beta(\frac{1}{L_g}\Vert x_s^0(k) - x_s^* \Vert) - \alpha_\ell(\Vert x(k) - x_s^0(k) \Vert).
	\end{split}
\end{equation} 
We now exploit a $\mathcal{K}_\infty$ function $$\alpha_{\beta \ell}(\Vert x \Vert) \coloneqq \min \left\{ \frac{1}{2}\alpha_\beta \left(\frac{1}{L_g}\Vert x \Vert \right),\alpha_\ell(\Vert x \Vert) \right\}$$ to further develop \eqref{eq:rnmpct_W_decrease_xs*_4} as
 \begin{equation}
	\label{eq:rnmpct_W_decrease_xs*_5}
	\begin{split}
		&V_{N_p}(\hat{x}^0(1 \vert k),y_t,\tilde{\mathbf{v}}(k+1),\tilde{y}_s(k+1)) - V^0_{N_p}(x(k),y_t)\\
		&\leq -\alpha_{\beta \ell}(\Vert x_s^0(k) - x_s^* \Vert) - \alpha_{\beta \ell}(\Vert x(k) - x_s^0(k) \Vert)\\
		&\leq -\alpha_{\beta \ell} \left( \frac{1}{2}\Vert x_s^0(k) - x_s^* \Vert + \frac{1}{2} \Vert x(k) - x_s^0(k) \Vert \right)\\
		&\leq -\alpha_{\beta \ell} \left( \frac{1}{2} \Vert x(k) - x_s^* \Vert \right).
	\end{split}
\end{equation} 
Exploiting a $\mathcal{K}_\infty$ function $\alpha_2(\Vert x \Vert) \leq \alpha_{\beta \ell}\left( \frac{1}{2}\Vert x \Vert \right)$, we finally obtain
 \begin{equation}
	\label{eq:rnmpct_W_decrease_xs*_6}
	\begin{split}
		&V_{N_p}(\hat{x}^0(1 \vert k),y_t,\tilde{\mathbf{v}}(k+1),\tilde{y}_s(k+1)) - V^0_{N_p}(x(k),y_t)\\ &\leq -\alpha_2 \left( \Vert x(k) - x_s^* \Vert \right),
	\end{split}
\end{equation} 
which proves the nominal descent of $W(\cdot,y_t)$. 

We now need to analyze $V_{N_p}(x(k+1),y_t,\tilde{\mathbf{v}}(k+1),\allowbreak \tilde{y}_s(k+1)) - V_{N_p}(\hat{x}^0(1 \vert k),y_t,\tilde{\mathbf{v}}(k+1),\tilde{y}_s(k+1))$, namely the contribution of the perturbation $w(k)$.
First, we prove that $(\tilde{\mathbf{v}}(k+1),\tilde{y}_s(k+1))$ is a feasible solution for $P_{N_p}(x(k+1),y_t)$. As $\hat{x}^0(1 \vert k) \in \mathcal{X}_f(\tilde{y}_s(k+1))$ and $\tilde{\mathbf{v}}(k+1) = \{ v_f(\tilde{y}_s(k+1)), \hdots, v_f(\tilde{y}_s(k+1)) \}$, we have that  $\phi_\pi(j;\tilde{y}_s(k+1),\hat{x}^0(1 \vert k),\tilde{\mathbf{v}}(k+1),\mathbf{0}) \in \mathcal{X}_f(\tilde{y}_s(k+1))$, for all $j = 1,\hdots,N_p$. Therefore, given the definition of the set sequence $\{\mathcal{F}(j)\}_{j \geq 0}$, we have that $\phi_\pi(j;\tilde{y}_s(k+1),x(k+1),\tilde{\mathbf{v}}(k+1),\mathbf{0}) \in \mathcal{X}_f(\tilde{y}_s(k+1)) \oplus \mathcal{F}(j)$. The monotonicity property of the sequence $\{\mathcal{R}(j)\}_{j \geq 0}$, which is used to build the restricted constraints $\mathcal{Z}_\pi(j)$, guarantees that $(\mathcal{X}_f(\tilde{y}_s(k+1)) \oplus \mathcal{F}(j) \times \{0\}) \subseteq \mathcal{Z}_\pi(j)$, for all $j = 1,\hdots,N_p$. Moreover, given that $\mathcal{X}_f(\tilde{y}_s(k+1)) \oplus \mathcal{F}(N_p-1) \subseteq \Omega(\tilde{y}_s(k+1))$, it also holds true that $\phi_\pi(N_p;\tilde{y}_s(k+1),x(k+1),\tilde{\mathbf{v}}(k+1),\mathbf{0}) \in \mathcal{X}_f(\tilde{y}_s)$, which proves that the constraints are verified and the last prediction, together with $\tilde{y}_s(k+1)$, is inside the robust invariant set for tracking $\mathcal{T}$.

Let us now analyze $V_{N_p}(x(k+1),y_t,\tilde{\mathbf{v}}(k+1),\tilde{y}_s(k+1)) - V_{N_p}(\hat{x}^0(1 \vert k), \allowbreak y_t,\tilde{\mathbf{v}}(k+1),\tilde{y}_s(k+1))$. We define $\hat{x}^+(j \vert k+1) \coloneqq \phi_\pi(j;\tilde{y}_s(k+1),x(k+1),\tilde{\mathbf{v}}(k+1),\mathbf{0})$ and $\hat{x}_n^+(j \vert k+1)    \coloneqq   \phi_\pi(j;\tilde{y}_s(k+1),\hat{x}^0(1 \vert k),\allowbreak \tilde{\mathbf{v}}(k+1),\mathbf{0})$. Then,
 \begin{equation}
	\begin{split}
		&V_{N_p}(x(k+1),y_t,\tilde{\mathbf{v}}(k+1),\tilde{y}_s(k+1))\\ &- V_{N_p}(\hat{x}^0(1 \vert k),y_t,\tilde{\mathbf{v}}(k+1),\tilde{y}_s(k+1))\\
		&\leq \sum_{j = 0}^{N_p-1} \sigma_{\ell,x} \left( \Vert \hat{x}^+(j \vert k+1) - \hat{x}_n^+(j \vert k+1) \Vert \right)\\ &+ \sigma_f \left( \Vert \hat{x}^+(N_p \vert k+1) - \hat{x}_n^+(N_p \vert k+1) \Vert \right) \\
		&\leq \sum_{j = 0}^{N_p-1} \sigma_{\ell,x} \circ r_j(w(k)) + \sigma_f \circ r_{N_p}(w(k)).
	\end{split}
\end{equation} 
Using a $\mathcal{K}$ function $\lambda_2(\Vert w \Vert) \geq \sum_{j = 0}^{N_p-1} \sigma_{\ell,x} \circ r_j(w) + \sigma_f \circ r_{N_p}(w)$ in $\mathcal{W}$, the last result and \eqref{eq:rnmpct_W_decrease_xs*_6} prove that
 \begin{equation}
	\label{eq:rnmpct_W_decrease_xs*_7}
	\Delta W \leq - \alpha_2(\Vert x(k) - x_s^* \Vert) + \lambda_2(\Vert w(k) \Vert).
\end{equation} 
Equations \eqref{eq:rnmpct_W_decrease_case_1_b} and \eqref{eq:rnmpct_W_decrease_xs*_7} respectively prove the nominal descent property of $W(\cdot,y_t)$ when $\alpha_\ell(\Vert x(k) - x_s^0(k) \Vert) \geq \gamma \alpha_O(\Vert y_s^0(k) - y_s^* \Vert)$ and when $\alpha_\ell(\Vert x(k) - x_s^0(k) \Vert) \leq \gamma \alpha_O(\Vert y_s^0(k) - y_s^* \Vert)$. Defining the $\mathcal{K}$ function $\lambda(\Vert w \Vert) \coloneqq \max \{ \lambda_1(\Vert w \Vert), \allowbreak \lambda_2(\Vert w \Vert) \}$ and the $\mathcal{K}_\infty$ function $\alpha(\cdot) \coloneqq \min \left\{ \underline{\alpha}(\cdot), \alpha_2(\cdot) \right\}$, with $\gamma = \min \left\{ \gamma_1,\gamma_2,\gamma_3 \right\}$, we can unify the two results and get
 \begin{equation}
	\Delta W \leq \lambda(\Vert w(k) \Vert) - \alpha (\Vert x(k) - x_s^* \Vert),
\end{equation} 
which, together with the lower and upper bounds defined in \eqref{eq:rnmpct_W_lower_bound_xs*} and \eqref{eq:rnmpct_W_upper_bound_xs*}, proves that $W(\cdot,y_t)$ is an ISS-Lyapunov function with respect to $x_s^*$, and that $x_s^*$ is input-to-state stable for the closed-loop system.

Input-to-state stability of $x_s^*$ and boundedness of $\mathcal{W}$ imply that, for any $x(0) \in \mathcal{X}_{N_p}$ and for any perturbation signal $\mathbf{w} \in \mathcal{W}^j$, there exist a $\mathcal{K}$ function $\gamma_a(\cdot)$ and a constant $\delta > 0$ such that
 \begin{equation}
	\lim_{k \to \infty} \sup \Vert x(k) - x_s^*\Vert \leq \gamma_a \left( \lim_{k \to \infty} \Vert w(k) \Vert \right) \leq \delta,
\end{equation} 
which proves that the closed-loop system state $x$ converges to the set $x_s^* \oplus \{x: \; \Vert x \Vert \leq \delta \}$. Given the uniform continuity of the control law $\pi(y_s,x,v)$ with respect to $(x,v)$, and the continuity of the output function $h(x,u)$ in all its arguments, this last finding implies that also the system output $y$ converges to a bounded set containing $y_s^*$, the best admissible steady output, which ends the proof.
\begin{remark}
	For large values of $w(k)$ it might still happen that the optimal artificial reference at instant $k$ is exactly $y_s^*$ and temporarily changes to a value $y_s^0(k+1) \neq y_s^*$ at instant $k+1$. Given that the feedback policies $\pi(y_s,x,v)$ are not required to be uniformly continuous with respect to $y_s$, this consideration implies that the set where the system output converges is bounded but might be non-connected. Choosing control laws $\pi(y_s,x,v)$ that are uniformly continuous also with respect to $y_s$ would ensure that the set where the system output converges is also connected.
\end{remark}

\subsection{Proof of the existence of the $\mathcal{K}_\infty$ function $\alpha_\beta(\cdot)$}
We want to prove that it is always possible to choose a feasible $\beta \in (\max \{ \beta_1,\beta_2 \},1)$ such that the function $g(\beta,s)$ can be lower-bounded by a $\mathcal{K}_\infty$ function $\alpha_\beta(s)$. In order for this to be true, it must hold true that, for $\beta \in (\max \{ \beta_1,\beta_2 \},1)$, $g(\beta,0) = 0$, $g(\beta,s) > 0$ for all $s > 0$, and $g(\beta,s_1) < g(\beta,s_2)$ if $s_1 < s_2$.

The first condition is easily proved by noting that in $s=0$ we get $s^2=0$ and $\alpha_O(s) = 0$. Since $0 < \beta < 1$, the second condition rewrites as
 \begin{equation}
	\alpha_O(s) - (1 - \beta)4\bar{\alpha}_f L_g^2 s_2 > 0, \; \forall s > 0.
\end{equation} 
This condition is verified if $\beta$ is such that $\beta > 1 - \frac{\alpha_O(s)}{4\bar{\alpha}_f L_g^2 s^2}$.
Given that $\beta > 1 - \mathfrak{b}_1$, this condition is verified for all $s \in \mathcal{S} \setminus \{0\}$.

Since $0 < \beta < 1$, the third condition requires to check that, for $s_2 > s_1$,
 \begin{equation}
	\alpha_O(s_1)   -   (1 - \beta)4\bar{\alpha}_f L_g^2 s_1^2   -   \alpha_O(s_2)   +   (1 - \beta)4\bar{\alpha}_f L_g^2 s_2^2 < 0, 
\end{equation} 
which reduces to
 \begin{equation}
	\beta > 1 - \frac{\alpha_O(s_2) - \alpha_O(s_1)}{4\bar{\alpha}_fs_2^2 - 4\bar{\alpha}_fs_1^2}.
\end{equation} 
As $\beta > 1 - \mathfrak{b}_2$, also this condition is verified for all $s_1,s_2 \in \mathcal{S}$ such that $s_2 > s_1$. In order to prove that $g(\beta,s)$ is lower-bounded by a $\mathcal{K}_\infty$ function, we must also ensure that $g(\beta,s) \to \infty$ for $s \to \infty$. However, as $\mathcal{S}$ is an interval, an extension of $g(\beta,s)$ can be chosen to grow arbitrarily fast for $s \in \mathbb{R} \setminus \mathcal{S}$. As a consequence, there exists a $\mathcal{K}_\infty$ function $\alpha_\beta(\cdot)$ such that $g(\beta,s) \geq \alpha_\beta(s)$.

\section{Supporting algorithms for the design of the controller} \label{sec:app_algorithms}
The computation of all the ingredients of the proposed MPC is not trivial and has to be carried out iteratively. This is because there is a strong interdependence between all the ingredients that does not allow us to build a simple sequence of steps to be followed. Indeed, to give an example, the computation of the terminal ingredients requires to know $\mathcal{F}(N_p-1)$, but the computation of $\mathcal{F}(N_p-1)$ requires to know the control laws $\pi(\cdot,\cdot,\cdot)$, that should be designed together with the terminal ingredients, given that $V_f(\cdot,\cdot)$ must behave as a CLF under the control law $\pi(y_s,x,v_f(y_s))$. Therefore, in the following we provide several algorithms that are needed to compute some of the ingredients and give some remarks that are needed to support the designer in the most critical steps.
In any case, we suggest to follow this simplified process:
\begin{enumerate}
	\item Design the pre-stabilizing control laws $\pi(\cdot,\cdot,\cdot)$ and the terminal ingredients using a candidate $\mathcal{Y}_t$ set.
	\item Compute the sequences $\{ \mathcal{F}(j) \}_{j \geq 0}$ and $\{ \mathcal{R}(j) \}_{j \geq 0}$ related to the designed laws $\pi(\cdot,\cdot,\cdot)$.
	\item Based on the found sequences, choose a prediction horizon $N_p$ verifying the assumptions and check that the terminal ingredients and the set $\mathcal{Y}_t$ computed at point 1 verify the assumptions, too. If not, go back to point 1.
\end{enumerate}

\subsection{Computation of the local control laws and the terminal ingredients}
\label{alg:terminal_ingredients}
The method we propose for the computation of the local control laws $\pi(y_s,x,v)$ and the terminal ingredients of our robust MPC for tracking is a slightly modified version of the method presented in \cite[Appendix~B]{Limon2018_NMPCT}, and based on \cite{Wan2003ltv,Wan2003auto}, for the computation of the terminal ingredients of a non-robust MPC for tracking. In particular, we consider the case of quadratic stage cost and terminal cost functions, for the sake of simplicity. 

The basic idea behind the method proposed in \cite[Appendix~B]{Limon2018_NMPCT} can be summed up in the following steps:
\begin{itemize}
	\item The chosen set of admissible setpoints $\mathcal{Y}_t$ is divided into smaller subsets $\mathcal{Y}_{s_i}$.
	\item Given $n_i$ samples $y_{s_j} \in \mathcal{Y}_{s_i}$, within the regions given by $(x_{s_j},v_{s_j})$ such that $x_{s_j} = g_x(y_{s_j})$ and $v_{s_j} = g_v(y_{s_j})$, for all $j \in [1,n_i]$, the nonlinear system is approximated as an LTV system, i.e.
	 \begin{equation}
		\begin{split}
			&f(x, u) = f(x_s(y_s), u_s(y_s))\\ &+ \sum_{j=1}^{n_i} \lambda_j \left[ A_{i,j} (x \! - \! g_x(y_s)) \! + B_{i,j} (u \! - \! g_u(y_s)) \right],
		\end{split}
	\end{equation} 
	where \( [A_{i,j}, B_{i,j}] \in \{ [A_{i,1}, B_{i,1}], \ldots, [A_{i,n_i}, B_{i,n_i}] \} \), \( \lambda_j \in [0, 1] \), and \( \sum_{j=1}^{n_i} \lambda_j = 1 \).
	\item Matrices $K_i$ and $P_i$ verifying the Lyapunov inequality 
	\[
	A_{K,i,j}^\top P_i A_{K,i,j} - P_i \leq -Q - K_i^\top R K_i, \quad \forall j,
	\]
	where \( A_{K,i,j} = A_{i,j} + B_{i,j} K_i \), are found by solving LMIs.
	\item Invariant sets $\Omega_i$ based on the control laws $\kappa(x,y_s) = K_i(x - g_x(y_s)) + g_v(y_s)$ are built for each subset $\mathcal{Y}_{s_i}$.
	\item An invariant set for tracking is obtained as 
	\[
	\Gamma = \bigcup_{i} \Gamma_i, \quad \Gamma_i = \left\{ (x, y_s) : x \in g_x(y_s) \oplus \Omega_i, y_s \in \mathcal{Y}_{s_i} \right\}.
	\]
\end{itemize}
In the robust context of the current work, two modifications to this method are made:
\begin{enumerate}
	\item We set the pre-stabilizing control laws as $\pi(y_s,x,v) = K_i (x - g_x(y_s)) + v$, $\forall y_s \in \mathcal{Y}_{s_i}$ and employ $v_f(y_s) = g_v(y_s)$ as terminal control inputs.
	\item Given that we need to find regions $\Omega(y_s)$ and $\mathcal{X}_f(y_s)$ such that $f \left( \hat{x},\pi \left( y_s,\hat{x},v_f(y_s) \right),0 \right) \in \mathcal{X}_f(y_s)$, for all $\hat{x} \in \Omega(y_s)$, we require the found matrices $K_i$ and $P_i$ to cause the contraction of the terminal cost on the nominal system, i.e.
	 \begin{equation}
		A_{K,i,j}^\top P_i A_{K,i,j} \leq \zeta P_i,
	\end{equation} 
	with $\zeta \in (0,1)$. Then, for all $y_s \in \mathcal{Y}_{s_i}$ and for all $x \in \Omega(y_s)$, we check that $f \left( \hat{x},\pi \left( y_s,\hat{x},v_f(y_s) \right),0 \right) \in \mathcal{X}_f(y_s)$. If the last condition is not satisfied, we try again with a new contraction factor $\zeta$.
\end{enumerate}
\begin{remark}
	The described method assumes to know $\mathcal{Y}_t$, even though its computation is possible only after having built the set sequences $\{ \mathcal{F}(j) \}_{j \geq 0}$ and $\{ \mathcal{R}(j) \}_{j \geq 0}$. We suggest to follow this algorithm using a (potential) outer approximation of the final $\mathcal{Y}_t$. Then, should the assumed outer approximation turn out not to verify the assumptions, the designer can start again the design process with a new $\mathcal{Y}_t$ candidate.
\end{remark}

\subsection{Computation of the constants $L_{x,ia}$, $L_{v,ib}$ and $L_{w,ic}$ for component-wise Lipschitz continuous systems}
\label{alg:lipschitz}
In this subsection, we present a method to compute constants $L_{x,ia}$, $L_{v,ib}$ and $L_{w,ic}$ for component-wise Lipschitz continuous systems such that, for all $i \in [1,n]$,
 \begin{equation}
	\label{eq:rnmpct_componentwise_Lipschitz}
	\begin{split}
		&\left\vert f_i(x,\pi(y_s,x,v),w) - f_i(\check{x},\pi(y_s,\check{x},\check{v}),\check{w}) \right\vert\\ &\leq \! \sum_{a = 1}^n L_{x,ia}\vert x_a \! - \! \check{x}_a \vert \! + \! \sum_{b = 1}^m L_{v,ib} \vert v_b \! - \! \check{v}_b \vert\\ &+ \! \sum_{c = 1}^r L_{w,ic} \vert w_c \! - \! \check{w}_c \vert.
	\end{split}
\end{equation} 
is verified.

For each $i \in [1,n]$, provided a candidate tuple $(L_{x,i1},\hdots,L_{x,in},L_{v,i1},\hdots,L_{v,im},L_{w,i1},\hdots,L_{w,ir})$, solve the optimization problem 
 \begin{equation}                                                                                                  
	\begin{split}                                                                                                    
		&\max_{ (x,v,w),(\check{x},\check{v},\check{w}),y_s} \!\!\!\!\!\!\! e_i \coloneqq \left\vert f_i(x,\pi(y_s,x,v),w) \! -\! f_i(\check{x},\pi(y_s,\check{x},\check{v}),\check{w}) \right\vert\\ &- \! \sum_{a = 1}^n L_{x,ia}\vert x_a \! - \! \check{x}_a \vert \! - \! \sum_{b = 1}^m L_{v,ib} \vert v_b \! - \! \check{v}_b \vert \! - \! \sum_{c = 1}^r L_{w,ic} \vert w_c \! - \! \check{w}_c \vert\\
		&\mathrm{s.t.:}\\
		&(x,\pi(y_s,x,v),w),(\check{x},\pi(y_s,\check{x},\check{v}),\check{w}) \in \mathcal{Z} \times \mathcal{W}, \; y_s \in \mathcal{Y}_t.
	\end{split}
\end{equation} 
If $e_i \leq 0$, then the candidate tuple verifies \eqref{eq:rnmpct_componentwise_Lipschitz}. Otherwise, choose a different candidate tuple and solve the previous optimization problem again. 
\begin{remark}
	Given that it is important to find small component-wise Lipschit constants, we suggest to start from large values for all of them, so as to verify \eqref{eq:rnmpct_componentwise_Lipschitz}, and then start gradually decreasing the value of each constant till the condition $e_i \leq 0$ is lost.
\end{remark}

\subsection{Computation of $\mathcal{F}(j)$ and $\mathcal{R}(j)$}
\label{alg:F_R}
In this subsection, we present a method to compute the sequences $\left\{ \mathcal{F}(j) \right\}_{j \geq 0}$ and $\left\{ \mathcal{R}(j) \right\}_{j \geq 0}$ for component-wise Lipschitz systems verifying \eqref{eq:rnmpct_componentwise_Lipschitz}.

In the component-wise Lipschitz case, the previously introduced functions $c_{i,j}(w)$ can be redefined as follows.
For all $i = 1,\hdots,n$, define $c_{i,0}(w) \coloneqq \sum_{c = 1}^r L_{w,ic} \cdot \vert w_c \vert$. Then, for $j > 0$, $c_{i,j}(w) \coloneqq \sum_{a = 1}^n L_{x,ia} \cdot c_{a,j-1}(w)$.

Since $\left\{\mathcal{F}(j)\right\}_{j \geq 0}$ and $\left\{\mathcal{R}(j)\right\}_{j \geq 0}$ have to take into account all the possible realizations of $w \in \mathcal{W}$, we consider the worst case, i.e. a realization $\breve{w}$ such that $\vert \breve{w}_i \vert = \bar{w}_i$, for all $i \in [1,r]$. Therefore, for all $j \geq 0$, $\mathcal{F}(j)$ can be obtained as
 \begin{equation}
	\label{eq:rnmpct_Lipschitz_F}
	\mathcal{F}(j) = \left\{ x \in \mathbb{R}^n: \; \vert x_i \vert \leq c_{i,j}(\breve{w}), \; \forall i \in [1,n] \right\}.
\end{equation}  
Now we introduce the functions $d_{i,j}(\cdot) \coloneqq \sum_{k = 0}^{j-1} c_{i,k}(\cdot)$, which can be used to compute $\mathcal{R}(j)$ as
 \begin{equation}
	\label{eq:rnmpct_Lipschitz_R}
	\mathcal{R}(j) = \left\{ x \in \mathbb{R}^n: \; \vert x_i \vert \leq d_{i,j}(\breve{w}), \; \forall i \in [1,n] \right\}.
\end{equation}  
Next, we present the lemma thanks to which the defined sequences $\left\{\mathcal{F}(j)\right\}_{j \geq 0}$ and $\left\{\mathcal{R}(j)\right\}_{j \geq 0}$ ensure that Assumptions \ref{ass:rnmpct_F_sequence} and \ref{ass:rnmpct_R_sequence} are verified.
\begin{lemma}[Consistency of the sequences $\left\{\mathcal{F}(j)\right\}_{j \geq 0}$ and $\left\{\mathcal{R}(j)\right\}_{j \geq 0}$ with Assumptions \ref{ass:rnmpct_F_sequence} and \ref{ass:rnmpct_R_sequence}]
	Given a system that verifies \eqref{eq:rnmpct_componentwise_Lipschitz}, the set sequences $\left\{\mathcal{F}(j)\right\}_{j \geq 0}$ and $\left\{\mathcal{R}(j)\right\}_{j \geq 0}$ computed as in \eqref{eq:rnmpct_Lipschitz_F} and \eqref{eq:rnmpct_Lipschitz_R} verify all the conditions of Assumptions \ref{ass:rnmpct_F_sequence} and \ref{ass:rnmpct_R_sequence}.
\end{lemma}
\begin{proof}
	This proof follows similar arguments as in \cite{Limon2010_robust_NMPC} and \cite{Manzano2019_Choki}.
	
	\textbf{Sequence $\left\{\mathcal{F}(j)\right\}_{j \geq 0}$:} Given that in the component-wise Lipschitz case, $\sigma_{w,ia}(w_a) = L_{w,ia}\cdot \vert w_a \vert$, defining 
	 \begin{equation}
		\mathcal{F}(0) = \Big\{ x \in \mathbb{R}^n: \; \vert x_i \vert \leq \sum_{a = 1}^r L_{w,ia} \vert \bar{w}_a \vert, \forall i \in [1,n] \Big\}
	\end{equation} 
	is consistent with the definition of $\mathcal{F}(0)$ given in Assumption \ref{ass:rnmpct_F_sequence}. Let us now analyze the gap between two nominal trajectories obtained starting from two initial conditions $x,\check{x}$ such that $x - \check{x} \in \mathcal{F}(0)$ and employing a feasible control sequence $\mathbf{v}$. We start by noting that, for all $i \in [1,n]$,
	 \begin{equation}
		\begin{split}
			&\vert \phi_{\pi,i}(1;x,y_s,\mathbf{v},\mathbf{0}) - \phi_{\pi,i}(1;\check{x},y_s,\mathbf{v},\mathbf{0}) \vert\\ &\leq \sum_{a = 1}^n L_{x,ia} \vert x_a - \check{x}_a \vert
			\leq \sum_{a = 1}^n L_{x,ia} c_{a,0}(\breve{w}) = c_{i,1}(\breve{w}),
		\end{split}
	\end{equation} 
	which proves that $\phi_{\pi}(1;\check{x},y_s,\mathbf{v},\mathbf{0}) \in \phi_{\pi}(1;x,y_s,\mathbf{v},\mathbf{0}) \oplus \mathcal{F}(1)$.
	We now analyze what happens for $j = 2$, i.e.
	 \begin{equation}
		\begin{split}
			&\vert \phi_{\pi,i}(2;x,y_s,\mathbf{v},\mathbf{0}) - \phi_{\pi,i}(2;\check{x},y_s,\mathbf{v},\mathbf{0}) \vert\\ &\leq \sum_{a = 1}^n L_{x,ia} \vert \phi_{\pi,a}(1;x,y_s,\mathbf{v},\mathbf{0}) - \phi_{\pi,a}(1;\check{x},y_s,\mathbf{v},\mathbf{0}) \vert\\
			&\leq \sum_{a = 1}^n L_{x,ia} c_{a,1}(\breve{w}) = c_{i,2}(\breve{w}),
		\end{split}
	\end{equation} 
	which proves that $\phi_{\pi}(2;\check{x},y_s,\mathbf{v},\mathbf{0}) \in \phi_{\pi}(2;x,y_s,\mathbf{v},\mathbf{0}) \oplus \mathcal{F}(2)$.
	
	Generalizing for all $j > 1$, we get that
	 \begin{equation}
		\begin{split}
			&\vert \phi_{\pi,i}(j;x,y_s,\mathbf{v},\mathbf{0}) - \phi_{\pi,i}(j;\check{x},y_s,\mathbf{v},\mathbf{0}) \vert\\ &\leq \sum_{a = 1}^n L_{x,ia} \vert \phi_{\pi,a}(j-1;x,y_s,\mathbf{v},\mathbf{0}) - \phi_{\pi,a}(j-1;\check{x},y_s,\mathbf{v},\mathbf{0}) \vert\\
			&\leq \sum_{a = 1}^n L_{x,ia} c_{a,j-1}(\breve{w}) = c_{i,j}(\breve{w}),
		\end{split}
	\end{equation} 
	allowing us to conclude that the sequence $\left\{\mathcal{F}(j)\right\}_{j \geq 0}$ verifies all the conditions of Assumption \ref{ass:rnmpct_F_sequence}.
	
	\textbf{Sequence $\left\{\mathcal{R}(j)\right\}_{j \geq 0}$:} Let us analyze the gap between a perturbation-free trajectory and a perturbed trajectory obtained starting from the same initial condition $x$ and employing the same feasible control sequence $\mathbf{v}$. We start by noting that, for all $i \in [1,n]$,
	 \begin{equation}
		\begin{split}
			&\vert \phi_{\pi,i}(1;x,y_s,\mathbf{v},\mathbf{w}) - \phi_{\pi,i}(1;x,y_s,\mathbf{v},\mathbf{0}) \vert\\ &\leq \sum_{a = 1}^r L_{w,ia} \vert \bar{w}_a \vert
			= c_{i,0}(\breve{w}) = d_{i,1}(\breve{w}),
		\end{split}
	\end{equation} 
	which proves that $\phi_{\pi}(1;x,y_s,\mathbf{v},\mathbf{w}) \in \phi_{\pi}(1;x,y_s,\mathbf{v},\mathbf{0}) \oplus \mathcal{R}(1)$. Then we analyze the gap after 2 steps, namely
	 \begin{equation}
		\begin{split}
			&\vert \phi_{\pi,i}(2;x,y_s,\mathbf{v},\mathbf{w}) - \phi_{\pi,i}(2;x,y_s,\mathbf{v},\mathbf{0}) \vert \leq \sum_{a = 1}^r L_{w,ia} \vert \bar{w}_a \vert\\
			&+ \sum_{b = 1}^n L_{x,ib} \vert \phi_{\pi,b}(1;x,y_s,\mathbf{v},\mathbf{w}) - \phi_{\pi,b}(1;x,y_s,\mathbf{v},\mathbf{0}) \vert\\
			&\leq c_{i,0}(\breve{w}) \! + \! \sum_{b = 1}^n L_{x,ib}c_{b,0}(\breve{w}) \! = \! c_{i,0}(\breve{w}) \! + \! c_{i,1}(\breve{w}) \! = \! d_{i,2}(\breve{w}),
		\end{split}
	\end{equation} 
	which proves that $\phi_{\pi}(2;x,y_s,\mathbf{v},\mathbf{w}) \in \phi_{\pi}(2;x,y_s,\mathbf{v},\mathbf{0}) \oplus \mathcal{R}(2)$. Finally, generalizing $\forall j > 1$,
	 \begin{equation}
		\begin{split}
			&\vert \phi_{\pi,i}(j;x,y_s,\mathbf{v},\mathbf{w}) - \phi_{\pi,i}(j;x,y_s,\mathbf{v},\mathbf{0}) \vert \leq \sum_{a = 1}^r L_{w,ia} \vert \bar{w}_a \vert\\
			&+ \sum_{b = 1}^n L_{x,ib} \vert \phi_{\pi,b}(j-1;x,y_s,\mathbf{v},\mathbf{w}) - \phi_{\pi,b}(j-1;x,y_s,\mathbf{v},\mathbf{0}) \vert\\
			&\leq \sum_{k = 0}^{j-1} c_{i,k}(\breve{w}) = d_{i,j}(\breve{w}),
		\end{split}
	\end{equation} 
	which allows us to state that the first condition of Assumption \ref{ass:rnmpct_R_sequence} is verified. We now turn our attention to the monotonicity of the sequence $\left\{\mathcal{R}(j)\right\}_{j \geq 0}$. Note that, as the sets $\mathcal{R}(j)$ and $\mathcal{F}(j)$ are box-shaped, we get that
	 \begin{equation}
		\mathcal{R}(j) \oplus \mathcal{F}(j) = \left\{ x \in \mathbb{R}^n: \! \vert x_i \vert \leq c_{i,j}(\breve{w}) \! + \! d_{i,j}(\breve{w}), \forall i \in [1,n] \right\},
	\end{equation} 
	but $c_{i,j}(\breve{w}) + d_{i,j}(\breve{w}) = d_{i,j+1}(\breve{w})$. Consequently, $\mathcal{R}(j) \oplus \mathcal{F}(j) = \mathcal{R}(j+1)$.
\end{proof}




\bibliography{rnmpct_tac_bibliography}


\end{document}